\documentclass[]{article}
\usepackage[nolinenums]{custom_arxiv}

\usepackage{chrkroer}

\usepackage[utf8]{inputenc} 
\usepackage[T1]{fontenc}    
\usepackage{url}            
\usepackage{booktabs}       
\usepackage{amsfonts}       
\usepackage{amsmath}
\usepackage{nicefrac}       
\usepackage{microtype}      
\usepackage{lipsum}
\usepackage{thm-restate}
\usepackage{wrapfig}
\usepackage{bm}

\usepackage{tikz}
\usepackage{graphicx}
\usetikzlibrary{fit,arrows,calc,positioning,shapes.geometric}
\tikzset{cross/.style={path picture={
      \draw[black] ($ (path picture bounding box.south east)!.25!(path picture bounding box.north west) $) -- ($ (path picture bounding box.south east)!.75!(path picture bounding box.north west) $);
    \draw[black] ($ (path picture bounding box.south west)!.25!(path picture bounding box.north east) $) -- ($ (path picture bounding box.south west)!.75!(path picture bounding box.north east) $);
  }}}

\usepackage{enumitem}
\usetikzlibrary{arrows,arrows.meta}
\newcommand{\dmid}{\parallel}

\DeclareRobustCommand{\treeplexproduct}{%
  \tikz[]{%
    \draw (0,0) circle[radius=1.1mm];%
    \draw (-.55mm, -.55mm) -- (.55mm, .55mm);%
    \draw (-.55mm, .55mm) -- (.55mm, -.55mm);%
  }%
}

\newcommand{\subt}[1]{\downarrow {#1}}
\newcommand{\seqf}[1]{\cX_{\subt{#1}}}
\newcommand{\xsubt}[1]{\vec{x}^t_{\subt{#1}}}
\newcommand{\symp}[1]{\Delta^{\!#1}}
\newcommand{\zsubt}[1]{\vec{z}^t_{\subt{#1}}}
\newcommand{\xhat}{{\hat x}}
\newcommand{\childinfosets}[1]{\mathcal{C}_{#1}}

\newcommand{\convsets}{\mathcal{J}}
\renewcommand{\vec}[1]{\bm{#1}}
\newcommand{\mat}[1]{\bm{#1}}

\title{Optimistic Regret Minimization for Extensive-Form Games via Dilated Distance-Generating Functions\thanks{This paper was accepted for publication at NeurIPS 2019.}}

\author{Gabriele Farina\\
Computer Science Department\\
Carnegie Mellon University\\
\texttt{gfarina@cs.cmu.edu}
\And
Christian Kroer\\
IEOR Department\\
Columbia University\\
\texttt{christian.kroer@columbia.edu}
\And
Tuomas Sandholm\\
Computer Science Department, CMU\\
Strategic Machine, Inc.\\
Strategy Robot, Inc.\\
Optimized Markets, Inc.\\
\texttt{sandholm@cs.cmu.edu}
}

\begin{document}
\maketitle

\begin{abstract}
  We study the performance of optimistic regret-minimization algorithms for both minimizing regret in, and computing Nash equilibria of, zero-sum extensive-form games. In order to apply these algorithms to extensive-form games, a distance-generating function is needed. We study the use of the dilated entropy and dilated Euclidean distance functions. For the dilated Euclidean distance function we prove the first explicit bounds on the strong-convexity parameter for general treeplexes. Furthermore, we show that the use of dilated distance-generating functions enable us to decompose the mirror descent algorithm, and its optimistic variant, into local mirror descent algorithms at each information set. This decomposition mirrors the structure of the counterfactual regret minimization framework, and enables important techniques in practice, such as distributed updates and pruning of cold parts of the game tree.
Our algorithms provably converge at a rate of $T^{-1}$, which is superior to prior counterfactual regret minimization algorithms. We experimentally compare to the popular algorithm CFR+, which has a theoretical convergence rate of $T^{-0.5}$ in theory, but is known to often converge at a rate of $T^{-1}$, or better, in practice. We give an example matrix game where CFR+ experimentally converges at a relatively slow rate of $T^{-0.74}$, whereas our optimistic methods converge faster than $T^{-1}$. We go on to show that our fast rate also holds in the Kuhn poker game, which is an extensive-form game. For games with deeper game trees however, we find that CFR+ is still faster. Finally we show that when the goal is minimizing regret, rather than computing a Nash equilibrium, our optimistic methods can outperform CFR+, even in deep game trees.
\end{abstract}

\section{Introduction}

Extensive-form games (EFGs) are a broad class of games that can model sequential interaction, imperfect information, and stochastic outcomes.
To operationalize them they must be accompanied by techniques for computing game-theoretic equilibria such as Nash equilibrium. A notable success story of this is poker: \citet{Bowling15:Heads} computed a near-optimal Nash equilibrium for heads-up limit Texas hold'em, while \citet{Brown17:Superhuman} beat top human specialist professionals at the larger game of heads-up no-limit Texas hold'em. Solving extremely large EFGs relies on many methods for dealing with the scale of the problem: abstraction methods are sometimes used to create smaller games~\citep{Gilpin07:Lossless, Lanctot12:No,Kroer14:Extensive,Ganzfried14:Potential,Brown15:Hierarchical,Kroer16:Imperfect}, endgame solving is used to compute refined solutions to the end of the game in real time~\citep{Burch14:Solving,Ganzfried15:Endgame,Moravcik16:Refining}, and recently depth-limited subgame solving has been very successfully used in real time~\citep{Moravvcik17:DeepStack,Brown18:Depth,Brown19:Superhuman}. At the core of all these methods is a reliance on a fast algorithm for computing approximate Nash equilibria of the abstraction, endgame, and/or depth-limited subgame~\cite{Moravvcik17:DeepStack,Brown18:Depth,Brown19:Superhuman}. In practice the most popular method has been the {\cfrp} algorithm~\citep{Zinkevich07:Regret,Tammelin15:Solving}, which was used within all three two-player poker breakthroughs~\citep{Bowling15:Heads,Moravvcik17:DeepStack,Brown17:Superhuman}. {\cfrp} has been shown to converge to a Nash equilibrium at a rate of $T^{-0.5}$, but in practice it often performs much better, even outperforming faster methods that have a guaranteed rate of $T^{-1}$~\citep{Brown17:Dynamic,Kroer18:Faster,Kroer18:Solving,Brown19:Solving}.

Recently, another class of optimization algorithms has been shown to have appealing theoretical properties. \emph{Online convex optimization} (OCO) algorithms are online variants of first-order methods: at each timestep $t$ they receive some loss function $\ell^t$ (often a linear loss which is a gradient of some underlying loss function), and must then recommend a point from some convex set based on the series of past points and losses. While these algorithms are generally known to have a $T^{-0.5}$ rate of convergence when solving static problems, a recent series of papers showed that when two \emph{optimistic} OCO algorithms are faced against each other, and they have some estimate of the next loss faced, a rate of $T^{-1}$ can be achieved~\citep{Rakhlin13:Online,Rakhlin13:Optimization,Syrgkanis15:Fast}.
In this paper we investigate the application of these algorithms to EFG solving, both in the regret-minimization setting, and for computing approximate Nash equilibria at the optimal rate of $O(T^{-1})$. The only prior attempt at using optimistic OCO algorithm in extensive-form games is due to~\citet{Farina19:Stable}. In that paper, the authors show that by restricting to the weaker notion of \emph{stable-predictive optimism}, one can mix and match local stable-predictive optimistic algorithm at every decision point in the game as desired and obtain an overall stable-predictive optimistic algorithm that enables $O(T^{-0.75})$ convergence to Nash equilibrium. The approach we adopt in this paper is different from that of~\citet{Farina19:Stable} in that our construction does not allow one to pick different regret minimizers for different decision points; however, our algorithms converge to Nash equilibrium at the improved rate $O(T^{-1})$.

The main hurdle to overcome is that in all known OCO algorithms a \emph{distance-generating function} (DGF) is needed to maintain feasibility via proximal operators and ensure that the stepsizes of the algorithms are appropriate for the convex set at hand.
\color{black}
For the case of EFGs, the convex set is known as a \emph{treeplex}, and the so-called dilated DGFs are known to have appealing properties, including closed-form iterate updates and strong convexity properties~\citep{Hoda10:Smoothing,Kroer18:Faster}. In particular, the dilated entropy DGF, which applies the negative entropy at each information set, is known to lead to the state-of-the-art theoretical rate on convergence for iterative methods~\citep{Kroer18:Faster}. Another potential DGF is the dilated Euclidean DGF, which applies the $\ell_2$ norm as a DGF at each information set. We show the first explicit bounds on the strong-convexity parameter for the dilated Euclidean DGF when applied to the strategy space of an EFG. We go on to show that when a dilated DGF is paired with the \emph{online mirror descent} (OMD) algorithm, or its optimistic variant, the resulting algorithm decomposes into a recursive application of local online mirror descent algorithms at each information set of the game. This decomposition is similar to the decomposition achieved in the counterfactual regret minimization framework, where a local regret minimizer is applied on the counterfactual regret at each information set. This localization of the updates along the tree structure enables further techniques, such as distributing the updates~\cite{Brown17:Superhuman,Brown15:Hierarchical} or skipping updates on cold parts of the game tree~\cite{Brown17:Reduced}.

It is well-known that the entropy DGF is the theoretically superior DGF when applied to optimization over a simplex~\citep{Hoda10:Smoothing}.
For the treeplex case where the entropy DGF is used at each information set,
\citet{Kroer18:Faster} showed that the strong theoretical properties of the
simplex entropy DGF generalize to the dilated entropy DGF on a treeplex (with earlier weaker results shown by \citet{Kroer15:Faster}). Our results on the dilated Euclidean DGF confirm this finding, as the dilated Euclidean DGF has a similar strong convexity parameter, but with respect to the $\ell_2$ norm, rather than the $\ell_1$ norm for dilated entropy (having strong convexity with respect to the $\ell_1$ norm leads to a tighter convergence-rate bound because it gives a smaller matrix norm, another important constant in the rate).

In contrast to these theoretical results, for the case of computing a Nash equilibrium in matrix games it has been found experimentally that the Euclidean DGF often performs much better than the entropy DGF. This was shown by \citet{Chambolle16:Ergodic} when using a particular accelerated primal-dual algorithm~\citep{Chambolle11:First,Chambolle16:Ergodic} and using the \emph{last iterate} (as opposed to the uniformly-averaged iterate as the theory suggests). \citet{Kroer19:First} recently showed that this extends to the theoretically-sound case of using linear or quadratic averaging in the same primal-dual algorithm, or in mirror prox~\citep{Nemirovski04:Prox} (the offline variant of optimistic OMD). In this paper we replicate these results when using OCO algorithms: first we show it on a particular matrix game, where we also exhibit a slow $T^{-0.74}$ convergence rate of {\cfrp} (the slowest {\cfrp} rate seen to the best of our knowledge). We show that for the Kuhn poker game the last iterate of optimistic OCO algorithms with the dilated Euclidean DGF also converges extremely fast. In contrast to this, we show that for deeper EFGs {\cfrp} is still faster. Finally we compare the performance of {\cfrp} and optimistic OCO algorithms for minimizing regret, where we find that OCO algorithms perform better.


\section{Regret Minimization Algorithms}

In this section we present the regret-minimization algorithms that we will work with. We will operate within the framework of \emph{online convex optimization}~\citep{Zinkevich03:Online}. In this setting, a decision maker repeatedly plays against an unknown environment by making decision $\vec{x}^1,\vec{x}^2,\ldots \in \cX$ for some convex compact set $\cX$. After each decision $\vec{x}^t$ at time $t$, the decision maker faces a \emph{linear loss} $\vec{x}^t \mapsto \langle \vec{\ell}^t, \vec{x}^t \rangle$, where $\vec{\ell}^t$ is a vector in $\cX$. Summarizing, the decision maker makes a decision $\vec{x}^{t+1}$ based on the sequence of losses $\vec{\ell}^1,\ldots,\vec{\ell}^t$ as well as the sequence of past iterates $\vec{x}^1,\ldots,\vec{x}^t$.

The quality metric for a regret minimizer is its \emph{cumulative regret}, which is the difference between the loss cumulated by the sequence
of decisions $\vec{x}^1, \dots, \vec{x}^T$ and the loss that would have been cumulated by playing the best-in-hindsight
time-independent decision $\hat{\vec{x}}$. Formally, the cumulative regret up to time $T$ is
\[
  R^T \defeq \sum_{t=1}^T \langle \vec{\ell}^t, \vec{x}^t\rangle - \min_{\hat{\vec{x}} \in \cX} \bigg\{ \sum_{t=1}^T \langle \vec{\ell}^t, \hat{\vec{x}} \rangle \bigg\}.
\]
A ``good'' regret minimizer is such that the cumulative regret grows \emph{sublinearly in $T$}.


The algorithms we consider assume access to a \emph{distance-generating function} $d: \cX \rightarrow \R$, which is 1-strongly convex (with respect to some norm) and continuously differentiable on the interior of $\cX$. Furthermore $d$ should be such that the gradient of the convex conjugate $\nabla d(\vec{g}) = \argmax_{\vec{x} \in \cX} \langle \vec{g}, \vec{x} \rangle - d(\vec{x})$ is easy to compute. Following \citet{Hoda10:Smoothing} we say that a DGF satisfying these properties is a \emph{nice} DGF for $\cX$.
From $d$ we also construct the \emph{Bregman divergence}
  $
    D(\vec{x}\dmid \vec{x}') \defeq d(\vec{x}) - d(\vec{x}') - \langle \nabla d(\vec{x}'), \vec{x} - \vec{x}' \rangle.
  $

First we present two classical regret minimization algorithms.
The \emph{online mirror descent} (OMD) algorithm produces iterates according to the rule
\begin{align}
  \label{eq:omd}
\vec{x}^{t+1} = \argmin_{\vec{x} \in \cX} \bigg\{ \langle \vec{\ell}^{t}, \vec{x} \rangle + \frac{1}{\eta} D(\vec{x} \dmid \vec{x}^{t})  \bigg\}.
\end{align}

The \emph{follow the regularized leader} (FTRL) algorithm produces iterates according to the rule~\citep{Schwartz07:Primal}
\begin{align}
  \label{eq:ftrl}
\vec{x}^{t+1} = \argmin_{\vec{x} \in \cX} \bigg\{ \bigg\langle \sum_{\tau=1}^{t} \vec{\ell}^\tau, \vec{x}\bigg\rangle + \frac{1}{\eta} d(\vec{x})  \bigg\}.
\end{align}

OMD and FTRL satisfy regret bounds of the form
$R^T \leq  O\left( D(\vec{x}^*\| \vec{x}^1) L \sqrt{T} \right)$
(e.g. \citet{Hazan16:Introduction}).


The \emph{optimistic} variants of the classical regret minimization algorithms take as input an additional vector $\vec{m}^{t+1}$, which is an estimate of the loss faced at time $t+1$~\citep{Chiang12:Online,Rakhlin13:Online}.
Optimistic OMD produces iterates according to the rule~\citep{Rakhlin13:Online} (note that $\vec{x}^{t+1}$ is produced before seeing $\vec{\ell}^{t+1}$, while $\vec{z}^{t+1}$ is produced after)
\begin{align}
  \label{eq:optimistic omd}
\vec{x}^{t+1} = \argmin_{\vec{x} \in \cX} \bigg\{\!  \langle \vec{m}^{t+1}, \vec{x} \rangle + \frac{1}{\eta} D(\vec{x} \dmid \vec{z}^{t}) \! \bigg\}, \ \
\vec{z}^{t+1} = \argmin_{\vec{z} \in \cX} \bigg\{ \langle \vec{\ell}^{t+1}, \vec{z}\rangle + \frac{1}{\eta} D(\vec{z}\dmid \vec{z}^{t})  \bigg\}.
\end{align}

Thus it is like OMD, except that $\vec{x}^{t+1}$ is generated by an additional step taken using the loss estimate. This additional step is transient in the sense that $\vec{x}^{t+1}$ is not used as a center for the next iterate.
OFTRL produces iterates according to the rule~\citep{Rakhlin13:Online,Syrgkanis15:Fast}
\begin{align}
  \label{eq:oftrl}
\vec{x}^{t+1} = \argmin_{\vec{x} \in \cX} \bigg\{ \bigg\langle \vec{m}^{t+1} + \sum_{\tau=1}^{t} \vec{\ell}^\tau, \vec{x}\bigg\rangle + \frac{1}{\eta} d(\vec{x})  \bigg\}.
\end{align}
Again the loss estimate is used in a transient way: it is used as if we already saw the loss at time $t+1$, but then discarded and not used in future iterations.

\subsection{Connection to Saddle Points}\label{sec:bspp}
A \emph{bilinear saddle-point problem} is a problem of the form
$
  \min_{\vec{x}\in \cX} \max_{\vec{y} \in \cY} \big\{ \vec{x}^{\!\top}\! \mat{A} \vec{y}\big\}, 
$%
where $\cX,\cY$ are closed convex sets.
This general formulation allows us to capture, among other settings, several game-theoretical applications such as computing Nash equilibria in two-player zero-sum games. In that setting, $\cX$ and $\cY$ are convex polytopes whose description is provided by the \emph{sequence-form constraints}, and $\mat{A}$ is a real payoff matrix~\citep{Stengel96:Efficient}.

The error metric that we use is the  \emph{saddle-point
  residual} (or \emph{gap}) $\xi$ of $(\bar{\vec{x}}, \bar{\vec{y}})$, defined as
$
  \xi(\bar{\vec{x}}, \bar{\vec{y}}) \defeq \max_{\hat{y}\in\cY} \langle \bar{\vec{x}}, \mat{A} \hat{\vec{y}}\rangle - \min_{\hat{\vec{x}}\in\cX} \langle \hat{\vec{x}}, \mat{A} \bar{\vec{y}} \rangle.
$
%
A well-known folk theorem shows that the average of a sequence of regret-minimizing strategies for the choice of losses
$
  \vec{\ell}^t_\cX : \cX \ni \vec{x} \mapsto  (-\mat{A}\vec{y}^{t})^{\!\top} x,\
  \vec{\ell}^t_\cY : \cY \ni \vec{y} \mapsto  (\mat{A}^\top \vec{x}^{t})^{\!\top} \vec{y}
$
leads to a bounded saddle-point residual, since
one has
 \begin{equation}\label{eq:gap regr sum}
   \xi(\bar{\vec{x}}, \bar{\vec{y}}) = \frac{1}{T}(R^T_\cX + R^T_\cY).
 \end{equation}

When $\cX,\cY$ are the players' sequence-form strategy spaces, 
this implies that
 the average strategy profile produced by the regret minimizers is a $\nicefrac{1}{T}(R^T_\cX + R^T_\cY)$-Nash equilibrium. This also implies that by using online mirror descent or follow-the-regularizer-leader, one obtains an anytime algorithm for computing a Nash equilibrium. In particular, at each time $T$, the average strategy output by each of the two regret minimizers forms a $\epsilon$-Nash equilibrium, where $\epsilon = O(T^{-0.5})$.

\subsection{RVU Property and Fast Convergence to Saddle Points}\label{sec:rvu}
Both optimistic OMD and optimistic FTRL satisfy the \emph{Regret bounded
by Variation in Utilities} (RVU) property, as given by~\citeauthor{Syrgkanis15:Fast}:

\begin{definition}[RVU property, \citep{Syrgkanis15:Fast}]\label{def:rvu}
     We say that a regret minimizer satisfies the RVU property if there exist constants $\alpha > 0$ and $0 < \beta \le \gamma$, as well as a pair of dual norms $(\|\cdot\|, \|\cdot\|_\ast)$ such that,     no matter what the loss functions $\vec{\ell}^1, \dots, \vec{\ell}^{T}$ are,
    \begin{equation}\label{eq:rvu}
    R^T \le \alpha + \beta \sum_{t=1}^T \|\vec{\ell}^t - \vec{m}^t\|_\ast^2 - \gamma \sum_{t=1}^T \|\vec{x}^t - \vec{x}^{t-1}\|^2. \tag{RVU}
    \end{equation}
\end{definition}

The definition given here is slightly more general than that of \citet{Syrgkanis15:Fast}: we allow a general estimate $\vec{m}^t$ of $\vec{\ell}^t$, whereas their definition requires using $\vec{m}^t = \vec{\ell}^{t-1}$. While the choice $\vec{m}^t = \vec{\ell}^{t-1}$ is often reasonable, in some cases other definitions of the loss prediction are more natural~\citep{Farina19:Stable}.
In practice, both optimistic OMD and optimistic FTRL satisfy a parametric notion of the RVU property, which depends on the value of the step-size parameter that was chosen to set up either algorithm.

\begin{theorem}[\citet{Syrgkanis15:Fast}]
  For all step-size parameters $\eta > 0$, Optimistic OMD satisfies the RVU conditions with respect to the primal-dual norm pair $(\|\cdot\|_1, \|\cdot\|_\infty)$ with parameters $\alpha = R/\eta, \beta = \eta, \gamma = 1/(8\eta)$, where $R$ is a constant that scales with the maximum allowed norm of any loss function $\ell$.
\end{theorem}

\begin{restatable}{theorem}{thmoftrl}\label{thm:oftrl}
  For all step-size parameters $\eta > 0$, OFTRL satisfies the RVU conditions with respect to any primal-dual norm pair $(\|\cdot\|, \|\cdot\|_\ast)$ with parameters $\alpha = \Delta_d/\eta, \beta = \eta, \gamma = 1/(4\eta)$, where $\Delta_d \defeq \max_{\vec{x},\vec{y}\in\cX} \{d(\vec{x}) - d(\vec{y})\}$.
\end{restatable}

Our proof, available in the appendix of the full paper, generalizes the work by~\citet{Syrgkanis15:Fast} by extending the proof beyond simplex domains and beyond the fixed choice $\vec{m}^t=\vec{\ell}^{t-1}$.

It turns out that this is enough to accelerate the convergence to a saddle point in the construction of Section~\ref{sec:bspp}. In particular, by letting the predictions be defined as
$
  \vec{m}^t_\cX \defeq \vec{\ell}_\cX^{t-1},
  \vec{m}^t_\cY \defeq \vec{\ell}_\cY^{t-1},
$
we obtain that the residual $\xi$ of the average decisions $(\bar{\vec{x}}, \bar{\vec{y}})$ satisfies
\begin{align*}
  T\xi(\bar{\vec{x}}, \bar{\vec{y}}) &\le \frac{2\alpha'}{\eta} + \eta\sum_{t=1}^T \bigg(\|{-\mat{A}}\vec{y}^t + \mat{A}\vec{y}^{t-1}\|_\ast^2 + \|\mat{A}^{\!\top}\!\vec{x}^t - \mat{A}^{\!\top}\! \vec{x}^{t-1}\|_\ast^2 \bigg) \\[-2mm]
  &\hspace{5cm}-\frac{\gamma\,'}{\eta}\sum_{t=1}^T \bigg(\|\vec{x}^t - \vec{x}^{t-1}\|^2
  + \|\vec{y}^t - \vec{y}^{t-1}\|^2 \bigg)\\[.5mm]
        &\leq  \frac{2\alpha'}{\eta} + \left(\eta \|\mat{A}\|_\text{op}^2 - \frac{\gamma\,'}{\eta} \right) \left(\sum_{t=1}^T \|\vec{x}^t - \vec{x}^{t-1}\|^2 + \sum_{t=1}^T \|\vec{y}^t - \vec{y}^{t-1}\|^2\right),
\end{align*}
where the first inequality holds by plugging~\eqref{eq:rvu} into~\eqref{eq:gap regr sum}, and the second inequality by noting that the operator norm $\|\cdot\|_\text{op}$ of a linear function is equal to the operator norm of its transpose. This implies that when the step-size parameter is chosen as $\eta = \frac{\sqrt{{\gamma\,'}}}{\|\mat{A}\|_\text{op}}$,
the saddle-point gap $\xi(\bar{\vec{x}}, \bar{\vec{y}})$ satisfies
$
  \xi(\bar{\vec{x}}, \bar{\vec{y}}) \le \frac{2 \alpha' \|\mat{A}\|_\text{op}}{T\sqrt{\gamma\,'}} = O(T^{-1}).
$


\section{Treeplexes and Sequence Form}\label{sec:treeplexes}



We formalize a sequential decision process as follows. We assume that we have a set of decision points $\convsets$. Each
decision point $j\in \convsets$ has a set of actions $A_j$ of size $n_j$. Given a
specific action at $j$, the set of possible decision points that the agent may
next face is denoted by $\childinfosets{j,a}$. It can be an empty set if no
more actions are taken after $j,a$. We assume that the decision points form a
tree, that is, $\childinfosets{j,a} \cap \childinfosets{j',a'} = \emptyset$ for all
other convex sets and action choices $j',a'$. This condition is equivalent to
the perfect-recall assumption in extensive-form games, and to conditioning on
the full sequence of actions and observations in a finite-horizon
partially-observable decision process.
In our
definition, the decision space starts with a root decision point, whereas in
practice multiple root decision points may be needed, for example in order to
model different starting hands in card games. Multiple root decision points can
be modeled by having a dummy root decision point with only a single action.

The set of possible next decision points after choosing action $a\in A_j$ at
decision point $j \in \convsets$, denoted $\childinfosets{j,a}$, can be
thought of as representing the different decision points that an agent may
face after taking action $a$ and then making an observation on which she can condition her next action choice.
In addition to games, our model of sequential decision process captures, for example, partially-observable Markov decision processes and Markov decision processes where we condition on the entire history of
observations and actions.


\begin{wrapfigure}{l}{6.5cm}
  \centering\includegraphics[scale=.6]{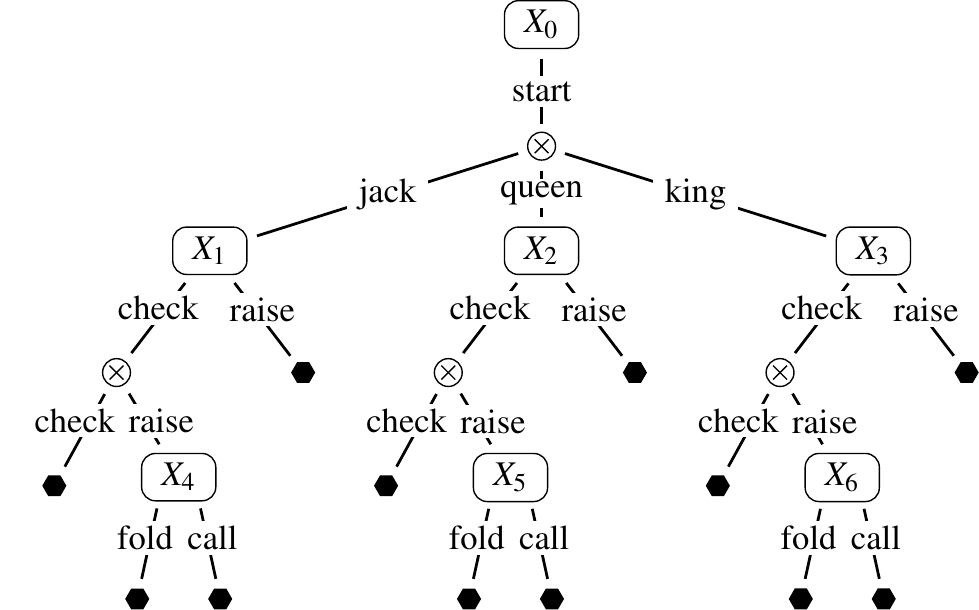}
  \caption{Sequential action space for the first player in the game of Kuhn poker. \treeplexproduct{} denotes an observation point; \protect\tikz{\protect\node[draw, circle, regular polygon, regular polygon sides=6, shape border rotate=180, inner sep=.7mm, fill=black] {}} represents the end of the decision process.}
  \label{fig:kuhn treeplex player1}
  \vspace{-2mm}
\end{wrapfigure}

As an illustration, consider the game of Kuhn poker~\citep{Kuhn50:Simplified}.
Kuhn poker consists of a three-card deck: king, queen, and jack.
The action space for the first player is shown in Figure~\ref{fig:kuhn treeplex
  player1}. For instance, we have: $\convsets = \{0,1,2,3,4,5,6\}$; $n_0 = 1$; $n_j = 2$ for all $j\in\convsets
\setminus \{0\}$; $A_0 = \{\text{start}\}$, $A_1 = A_2 = A_3 = \{\text{check},
\text{raise}\}$, $A_4 = A_5 = A_6 = \{\text{fold}, \text{call}\}$;
$\childinfosets{0,\text{start}} = \{1, 2, 3\}$, $\childinfosets{1, \text{raise}}
= \emptyset$, $\childinfosets{3,\text{check}} = \{6\}$; etc.

The expected loss for a given strategy is non-linear in the vectors of probability masses for each decision point $j$. This non-linearity is due to the probability of reaching each $j$, which is computed as the product of the probabilities
of all actions on the path to from the root to $j$. 
%
An alternative formulation which preserves linearity is called the \emph{sequence form}. In the sequence-form
representation, the simplex strategy space at a generic decision point $j\in \cJ$ is scaled by the decision variable associated with the last action in the path from the root of the process to $j$. In this formulation, the value of a particular action
represents the probability of playing the whole \emph{sequence} of actions from
the root to that action. This allows each term in the expected loss to be
weighted only by the sequence ending in the corresponding action. The sequence
form has been used to instantiate linear programming~\citep{Stengel96:Efficient}
and first-order methods~\citep{Hoda10:Smoothing,Kroer15:Faster,Kroer18:Faster}
for computing Nash equilibria of zero-sum EFGs.
Formally, the sequence-form representation $\cX $ of a sequential decision process can be obtained recursively, as follows:
for every $j\in \cJ$, $a \in A_j$, we let
    $
      \seqf{j,a} \defeq \prod_{j'\in\childinfosets{j,a}} \seqf{j'},
    $
    where $\Pi$ denotes Cartesian product;
  at every decision point $j \in \cJ$, we let
    \[
        \seqf{j} \defeq \{(
                                   \lambda_1 ,
                                   \dots ,
                                   \lambda_{n_j},
                                   \lambda_1 \vec{x}_{a_1},
                                   \dots,
                            \lambda_{n_j}\vec{x}_{a_{n_j}}): (\lambda_1, \dots, \lambda_n) \in \Delta^{n_j},\vec{x}_{a} \in \seqf{j,a} \ \forall\, a\in A_j \},\]
where we assumed $A_j = \{a_1, \dots, a_{n_j}\}$.

The sequence form strategy space for the whole sequential decision process is then $\cX \defeq \{1\} \times \seqf{r}$, where $r$ is the root of the process. The first entry, identically equal to 1 for any point in $\cX$, corresponds to what is called the \emph{empty sequence}. Crucially, $\cX $ is a convex and compact set, and the expected loss of the process is a linear function over $\cX $.
With the sequence-form representation the problem of computing a Nash
equilibrium in an EFG can be formulated as a \emph{bilinear saddle-point
  problem} (see Section~\ref{sec:bspp}), where $\cX$ and $\cY$ are the sequence-form strategy spaces of the sequential decision processes faced by the two players,
and $\mat{A}$ is a sparse matrix encoding the leaf payoffs of the game.

As we have already observed, vectors that pertain to the sequence form have one entry for each sequence of the decision process. We denote with $v_\phi$ the entry in $\vec{v}$ corresponding to the empty sequence, and $v_{ja}$ the entry corresponding to any other sequence $(j, a)$ where $j\in\cJ, a \in A_j$. Sometimes, we will need to \emph{slice} a vector $\vec{v}$ and isolate only those entries that refer to all decision points $j'$ and actions $a' \in A_{j'}$ that are at or below some $j\in \cJ$; we will denote such operation as $\vec{v}_{\downarrow j}$. Similarly, we introduce the syntax $v_j$ to denote the subset of $n_j = |A_j|$ entries of $\vec{v}$ that pertain to all actions $a \in A_j$ at decision point $j\in\cJ$.
Finally, note that for any $j\in\cJ-\{r\}$ there is a unique sequence $(j', a')$, denoted $p_j$ and called \emph{the parent sequence of decision point $j$}, such that $j \in \childinfosets{j'a'}$. When $j = r$ is the root decision point, we let $p_r \defeq \phi$, the empty sequence.

\section{Dilated Distance Generating Functions}

We will be interested in a particular type of DGF which is suitable for
sequential decision-making problems: a \emph{dilated DGF}. A dilated DGF is
constructed by taking a sum over suitable local DGFs for each decision point,
where each local DGF is dilated by the parent variable leading to the
decision point:
$
  d(\vec{x}) = \sum_{j \in \convsets} x_{p_j}\
d_j\!\bigg(\frac{\vec{x}_j}{x_{p_j}}\bigg).
$
Each ``local'' DGF $d_j$ is given the local variable $\vec{x}_j$ divided by $x_{p_j}$, so that $\frac{\vec{x}_j}{x_{p_j}} \in \symp{n_j}$. The idea is that $d_j$ can be any DGF suitable for $\symp{n_j}$; by multiplying $d_j$ by $x_{p_j}$ and taking a sum over $\convsets$ we construct a DGF for the whole treeplex from these local DGFs.
\citet{Hoda10:Smoothing} showed that dilated DGFs have many of the desired
properties of a DGF for an optimization problem over a treeplex.

We now present two local DGFs for simplexes, that are by far the most common in practice. In the following we let $\vec{b}$ be a vector in the $n$-dimensional simplex $\symp{n}$. First, the \emph{Euclidean DGF} $d(\vec{b}) = \|\vec{b}\|_2^2$, which is 1-strongly convex with respect to the $\ell_2$ norm; secondly, the \emph{negative entropy DGF} $d(\vec{b}) = \sum_{i=1}^n b_i \log(b_i)$ (we will henceforth drop the ``negative'' and simply refer to it as the entropy DGF), which is 1-strongly convex with respect to the $\ell_1$ norm.
%
The strong convexity properties of the dilated entropy DGF were shown by \citet{Kroer18:Faster} (with earlier weaker results shown by \citet{Kroer15:Faster}). However, for the dilated Euclidean DGF a setup for achieving a strong-convexity parameter of 1 was unknown until now; \citet{Hoda10:Smoothing} show that a strong-convexity parameter exists, but do not show what it is for the general case (they give specific results for a particular class of \emph{uniform treeplexes}). We now show how to achieve this.


We are now ready to state our first result on dilated regularizers that are strongly convex with respect to the Euclidean norm:

\begin{restatable}{theorem}{thmstronglconvexdgf}\label{thm:strongly convex dgf}
    Let $d(\vec{x}) = \sum_{j\in\cJ} x_{p_j} d_j(\vec{x}_j/x_{p_j})$ where for all $j$, $d_j$ is $\mu_j$-strongly convex with respect to the Euclidean norm over $\Delta^{n_j}$. Furthermore, define
$
  \sigma_{ja} := \frac{\mu_j}{2} - \sum_{j' \in C_{ja}} \mu_{j'},$ and $ \bar \sigma := \min_{ja} \sigma_{ja}.
$
Then, $d$ is $\bar\sigma$-strongly convex with respect to the Euclidean norm over $\cX$.
\end{restatable}

We can immediately use Theorem~\ref{thm:strongly convex dgf} to prove the following corollary:
\begin{corollary}\label{cor:dgf weights}
Let $\bar\sigma > 0$ be arbitrary, and for all $j$ let $d_j$ be a $\mu_j$-strongly convex function over $\symp{n_j}$ with respect to the Euclidean norm, where the $\mu_j$'s satisfy
\begin{equation}\label{eq:deja vu}
  \mu_j = 2\bar\sigma + 2\max_{a\in A_j}\sum_{j' \in \childinfosets{ja}} \mu_{j'}.
\end{equation}
Then, $d(\vec{x}) = \sum_{j\in\cJ} x_{p_j} d_j(\vec{x}_j/x_{p_j})$ is $\bar\sigma$-strongly convex over $\cX$ with respect to the Euclidean norm.
\end{corollary} 
\section{Local Regret Minimization}

We now show that OMD and Optimistic OMD run on a treeplex $\cX$ with a dilated DGF can both be interpreted as locally minimizing a modified variant of loss at each information set, with correspondingly-modified loss predictions. The modified local loss at a given information set $j$ takes into account the loss and DGF below $j$ by adding the expectation with respect to the next iterate $\xsubt{j}$. In practice this modified loss is easily handled by computing $\vec{x}^t$ bottom-up, thereby visiting $j$ after having visited the whole subtree below.

We first show that the problem of computing the \emph{prox mapping}, the minimizer of a linear term plus the Bregman divergence, decomposes into local prox mappings at each simplex of a treeplex. This will then be used to show that OMD and Optimistic OMD can be viewed as a tree of local simplex-instantiations of the respective algorithms.

\subsection{Decomposition into Local Prox Mappings with a Dilated DGF}
We will be interested in solving the following prox mapping, which takes
place in the sequence form:
\begin{align}
  \label{eq:sequence form prox}
  \prox(\vec{g},\hat{\vec{x}}) = \argmin_{\vec{x} \in \cX} \big\{ \langle \vec{g}, \vec{x} \rangle + D(\vec{x} \dmid \hat{\vec{x}}) \big\}.
\end{align}
The reason is that the update applied at each iteration of several OCO
algorithms run on the sequence-form polytope of $\cX$ can be described as an
instantiation of this prox mapping. We now show that this update can be
interpreted as a local prox mapping at each decision point, but with a new
loss $\hat g_j$ that depends on the update applied in the subtree beneath
$j$.

\begin{restatable}[Decomposition into local prox mappings]{proposition}{propdecompositionlocalprox}
  \label{prop:decomposition local prox}
  A prox mapping \eqref{eq:sequence form prox} on a treeplex with a
Bregman divergence constructed from a dilated DGF decomposes into local prox
mappings at each decision point $j$ where the solution is as follows:
\[
    \vec{x}^*_j = x_{p_j}\cdot \argmin_{\vec{b}_j \in \symp{n_j}} \left\{\left\langle \hat{\vec{g}}_{j}, \vec{b}_j\right\rangle + D_j\bigg(
\vec{b}_j\,
  \right\|\left.
\frac{\hat{\vec{x}}_j}{\xhat_{p_j}}\bigg)
  \right\},
\]
  where
\[
    \hat g_{j,a} = g_{j,a} + \sum_{j' \in \childinfosets{j,a}}
\left[
d_{\subt{j'}}^*\big(-\vec{g}_{\subt{j'}} + \nabla d_{\subt{j'}}(\hat{\vec{x}}_{\subt{j'}})\big) -
d_{j'}\bigg(\frac{\hat{\vec{x}}_j}{\xhat_{p_j}}\bigg) + \left\langle  \nabla
d_{j'}\!\!\left(\frac{\hat{\vec{x}}_{j'}}{\xhat_{p_{j'}}}\right), \frac{\hat{\vec{x}}_{j'}}{\xhat_{p_{j'}}} \right\rangle\right].
\]
\end{restatable}

\citet{Hoda10:Smoothing} and \citet{Kroer18:Solving} gave variations on a
similar result: that the convex conjugate $d_{\subt{j}}^*(-\vec{g})$ can be
computed in bottom-up fashion similar to the recursion we show here.
Proposition~\ref{prop:decomposition local prox} is slightly different in that
we additionally show that the Bregman divergence also survives the
decomposition and can be viewed as a local Bregman divergence. This latter
difference will be necessary for showing that OMD can be interpreted as a
local RM.

\subsection{Decomposition into Local Regret Minimizers}
With Proposition~\ref{prop:decomposition local prox} it follows almost directly that OMD and Optimistic OMD can be seen as a set of local regret minimizers, one for each simplex. Each produces iterates from their respective simplex, with the overall strategy produced by then applying the sequence-form transformation to these local iterates.

\begin{restatable}{theorem}{thmomdseparable}
  OMD with a dilated DGF for a treeplex $\cX$ corresponds to running OMD locally at each simplex $j$, with the local loss $\hat{\vec{\ell}}^t$ constructed according to Proposition~\ref{prop:decomposition local prox}.
  Optimistic OMD corresponds to the optimistic variant of this local OMD with local loss predictions $\hat{\vec{\ell}}^t, \hat{\vec{m}}_{j}^{t+1}$ again constructed according to Proposition~\ref{prop:decomposition local prox} using $\vec{x}^t$ as Bregman divergence center and $\vec{x}^{t+1}$ for aggregating losses below each simplex.
  Here the modified loss uses $\zsubt{j'}$ and $\vec{x}^{t+1}$ as Bregman divergence center and aggregating loss below, respectively. The prediction $\hat{\vec{m}}_{j}^{t+1}$ uses $\zsubt{j'}$ and $\vec{z}^{t+1}$.
\end{restatable}

Unlike OMD and its optimistic variant, it is not the case that FTRL has a nice interpretation as a local regret minimizer. The reason is that the prox mapping in \eqref{eq:ftrl} or \eqref{eq:oftrl} minimizes the sum of losses, rather than the most recent loss. Because of this, the expected value $\langle \sum_{\tau=1}^t\vec{\ell}_{\subt{j}}^\tau, \vec{x}_{\subt{j}}^{t+1} \rangle$ at simplex $j$, which influences the modified loss at parent simplexes, is computed based on $\vec{x}^{t+1}$ for all $t$ losses. Thus there is no local modified loss that could be received at rounds $1$ through $t$ that accurately reflects the modified loss needed in Proposition~\ref{prop:decomposition local prox}.


\section{Experimental Evaluation}

We experimentally evaluate the performance of optimistic regret minimization methods instantiated with dilated distance-generating functions. We experiment on three games:
\begin{itemize}[nolistsep,itemsep=1mm,leftmargin=*]
  \item \emph{Smallmatrix}, a small $2 \times 2$ matrix game. Given a mixed strategy $\vec{x} = ( x_1, x_2) \in \Delta^2$ for Player 1 and a mixed strategy $\vec{y} = (y_1, y_2)\in\Delta^2$ for Player 2, the payoff function for player 1 is
      $
        u(x, y) = 5 x_1 y_1 - x_1 y_2 + x_2 y_2.
      $
  \item Kuhn poker, already introduced in Section~\ref{sec:treeplexes}. In Kuhn poker, each player
first has to put a payment of 1 into the pot. Each player is then dealt one of the three cards, and the third is put
aside unseen. A single round of betting then occurs: first, Player $1$ can check or bet $1$. Then,
  \begin{itemize}[nolistsep]
  \item If Player $1$ checks Player $2$ can check or raise $1$.
    \begin{itemize}[nolistsep]
      \item If Player $2$ checks a showdown occurs; if Player $2$ raises Player $1$ can fold or call.
        \begin{itemize}
          \item If Player $1$ folds Player $2$ takes the pot; if Player $1$ calls a showdown occurs.
          \end{itemize}
        \end{itemize}
      \item If Player $1$ raises Player $2$ can fold or call.
        \begin{itemize}[nolistsep]
        \item If Player $2$ folds Player $1$ takes the pot; if Player $2$ calls a showdown occurs.
        \end{itemize}
      \end{itemize}
If no player has folded, a showdown occurs where the player with the higher card wins.
  \item Leduc poker, a standard benchmark in imperfect-information game solving~\cite{Southey05:Bayes}. The game is played with
a deck consisting of 5 unique cards with 2 copies of each, and consists of two rounds.
In the first round, each player places an ante of $1$ in the pot and receives a
single private card. A round of betting then takes place with a two-bet maximum,
with Player 1 going first. A public shared card is then dealt face up and
another round of betting takes place. Again, Player 1 goes first, and there is a
two-bet maximum. If one of the players has a pair with the public card, that
player wins. Otherwise, the player with the higher card wins. All bets in the first round are $1$, while all bets in the second round are $2$. This game has 390 decision points and 911 sequences per player.
\end{itemize}

\vspace{5mm}
\noindent\textbf{Fast Last-Iterate Convergence.} In the first set of experiments (Figure~\ref{fig:exp}, top row), we compare the saddle-point gap of the strategy profiles produced by optimistic OMD and optimistic FTRL to that produced by CFR and CFR$^+$. Optimistic OMD and optimistic FTRL were set up with the step-size parameter $\eta=0.1$ in Smallmatrix and $\eta=2$ in Kuhn Poker, and the plots show the last-iterate convergence for the optimistic algorithms, which has recently received attention in the works by~\citet{Chambolle16:Ergodic} and~\citet{Kroer19:First}. Finally, we instantiated optimistic OMD and optimistic FTRL with the Euclidean distance generating function as constructed in Corollary~\ref{cor:dgf weights}. The plots show that---at least in these shallow games---optimistic methods are able to produce even up to 12 orders of magnitude better-approximate saddle-points than CFR and CFR$^+$.

Interestingly, Smallmatrix appears to be a hard instance for CFR$^+$: linear regression on the first 20\,000 iterations of CFR$^+$ shows, with a coefficient of determination of roughly $0.96$, that
$  \log \xi(\vec{x}^T_*, \vec{y}^T_*) \approx -0.7375\cdot \log(T) -2.1349$,
where $(\vec{x}^T_*, \vec{y}^T_*)$ is the average strategy profile (computed using linear averaging, as per $CFR^+$'s construction) up to time $T$. In other words, we have evidence of at least one game in which the approximate saddle-point computed by CFR$^+$ experimentally has residual bounded below by $\Omega(T^{-0.74})$. This observation suggests that the analysis of CFR$^+$ might actually be quite tight, and that CFR$^+$ is \emph{not} an accelerated method.

Figure~\ref{fig:exp} (bottom left) shows the performance of OFTRL in Leduc Poker, compared to CFR and CFR$^+$ (we do not show optimistic OMD, which we found to have worse performance than OFTRL). Here OFTRL performs worse than CFR$^+$.
This shows that in deeper games, more work has to be done to fully exploit the accelerated bounds of optimistic regret minimization methods.

\begin{figure}[ht]\centering
    \includegraphics[scale=.74]{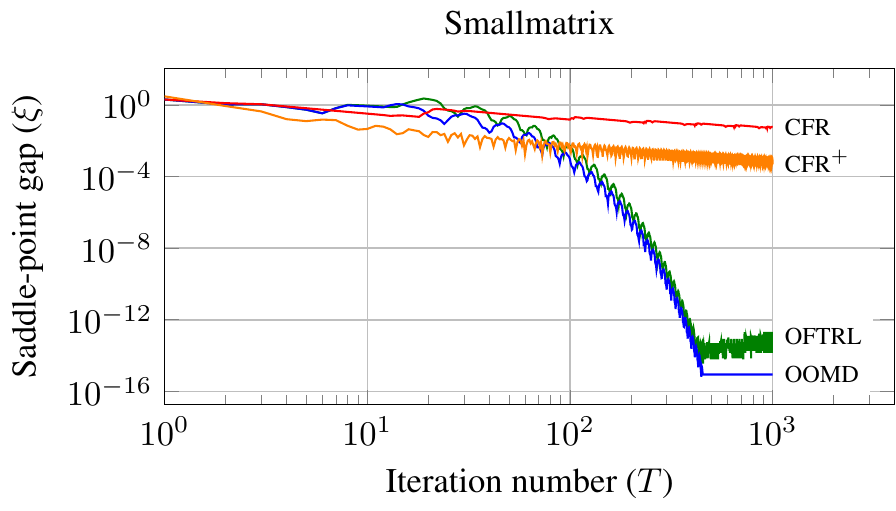}\hspace{1cm}
    \includegraphics[scale=.74]{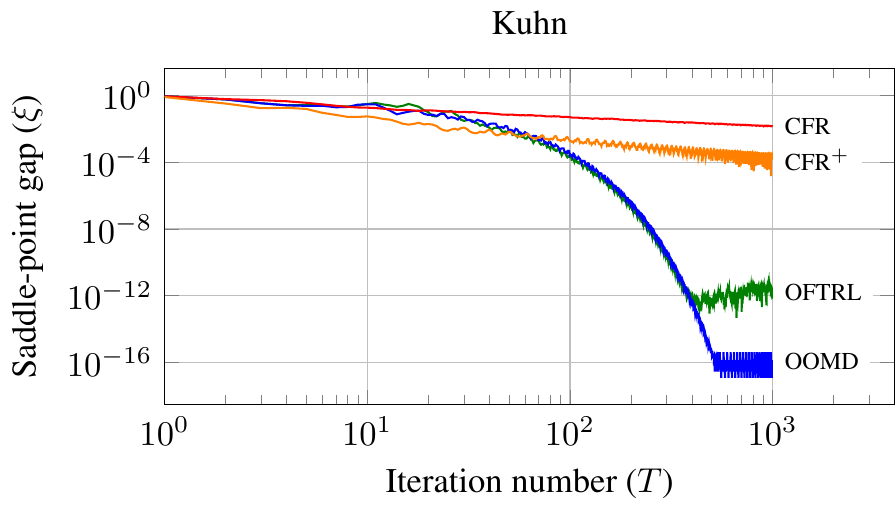}\\
        \includegraphics[scale=.74]{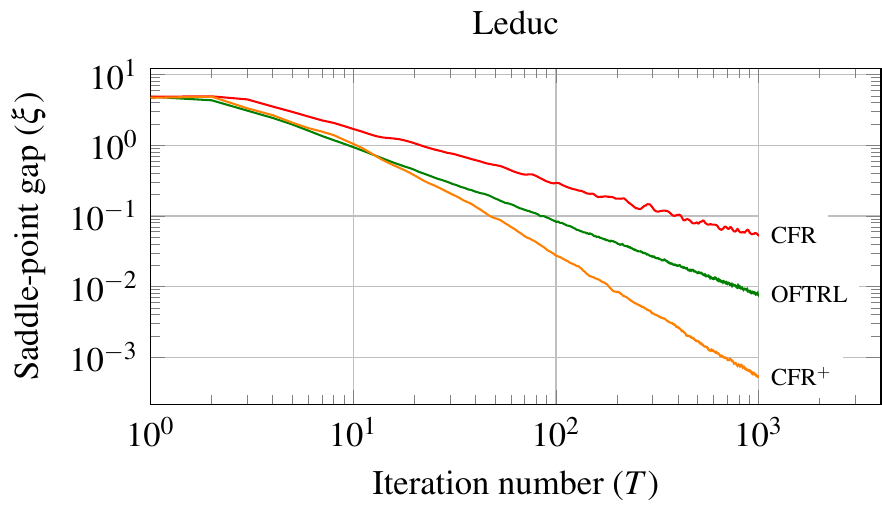}\hspace{1cm}
    \includegraphics[scale=.74]{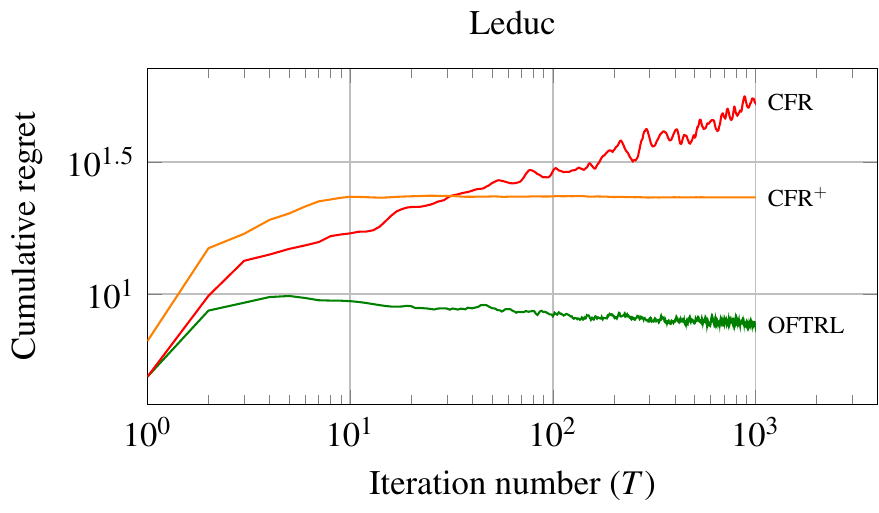}
    \caption{
    (Left and upper right) Saddle-point gap as a function of the number of iterations. The plots show the last-iterate convergence for OOMD and OFTRL.(Lower right) Sum of cumulative regret for both players in Leduc. Optimistic OMD (OOMD) and OFTRL use step-size parameter $\eta=0.1$ in Smallmatrix and $\eta=2$ in Kuhn. OFTRL uses step-size parameter $\eta=200$ in Leduc.
    }
    \label{fig:exp}
\end{figure}

\vspace{3mm}
\noindent\textbf{Comparing the Cumulative Regret.} We also compared the algorithms based on the sum of cumulative regrets (again we omit optimistic OMD, which performed worse than OFTRL). In all three games, OFTRL leads to lower sum of cumulative regrets. Figure~\ref{fig:exp} (bottom right) shows the performance of OFTRL in Leduc Poker. Here, we used the usual average of iterates $\bar{\vec{x}} \defeq \nicefrac{1}{T}\sum_{t=1}^T \vec{x}^t$ (note that the choice of averaging strategy has no effect on the bottom right plot.)

OFTRL's performance matches the theory from Theorem~\ref{thm:oftrl} and Section~\ref{sec:rvu}. In particular, we observe that while OFTRL does not beat the state-of-the-art CFR$^+$ in terms of saddle-point gap, it beats it according to the regret sum metric.
The fact that CFR$^+$ performs worse with respect to the regret sum metric is
somewhat surprising: the entire derivation of CFR and CFR$^+$ is based on
showing bounds on the regret sum. However, the connection between regret and
saddle-point gap (or exploitability) is one-way: if the two regret minimizers
(one per player) have regret $R_1$ and $R_2$, then the saddle point gap can be
easily shown to be less than or equal to $(R_1 + R_2)/T$. However, nothing prevents
it from being much smaller than $(R_1 + R_2)/T$. What we empirically find is that
for CFR$^+$ this bound is very loose. We are not sure why this is the
case, and it potentially warrants further investigation in the future.

\section{Conclusions}

We studied how optimistic regret minimization can be applied in the context of extensive-form games, and introduced the first instantiations of regret-based techniques that achieve $T^{-1}$ convergence to Nash equilibrium in extensive-form games. These methods rely crucially on having a tractable regularizer to maintain feasibility and control the stepsizes on the domain at hand---in our case, the sequence-form polytope. We provided the first explicit bound on the strong convexity properties of dilated distance-generating functions with respect to the Euclidean norm. We also showed that when optimistic regret minimization methods are instantiated with dilated distance-generating functions, the regret updates are local to each information set in the game, mirroring the structure of the counterfactual regret minimization framework. This localization of the updates along the tree structure enables further techniques, such as distributing the updates or skipping updates on cold parts of the game tree. Finally, when used in self play, these optimistic
regret minimization methods guarantee an optimal $T^{-1}$ convergence rate to Nash equilibrium.

We demonstrate that in shallow games, methods based on optimistic regret minimization can significantly outperform CFR and CFR$^+$---even up to 12 orders of magnitude. In deeper games, more work has to be done to fully exploit the accelerated bounds of optimistic regret minimization methods. However, while the strong CFR$^+$ performance in large games remains a mystery, we elucidate some points about its performance---including showing that its theoretically slow convergence bound is somewhat tight. Finally, we showed that when the goal is minimizing regret, rather than computing a Nash equilibrium, optimistic methods can outperform CFR$^+$ even in deep game trees.

\section*{Acknowledgments}
    This material is based on work supported by the National
    Science Foundation under grants IIS-1718457, IIS-1617590,
    and CCF-1733556, and the ARO under award W911NF-17-1-0082. Gabriele Farina is supported by a Facebook fellowship.

\bibliographystyle{custom_arxiv}
\bibliography{dairefs}

\begin{thebibliography}{38}
\providecommand{\natexlab}[1]{#1}
\providecommand{\url}[1]{\texttt{#1}}
\expandafter\ifx\csname urlstyle\endcsname\relax
  \providecommand{\doi}[1]{doi: #1}\else
  \providecommand{\doi}{doi: \begingroup \urlstyle{rm}\Url}\fi

\bibitem[Bowling et~al.(2015)Bowling, Burch, Johanson, and
  Tammelin]{Bowling15:Heads}
Bowling, M., Burch, N., Johanson, M., and Tammelin, O.
\newblock Heads-up limit hold'em poker is solved.
\newblock \emph{Science}, 347\penalty0 (6218), January 2015.

\bibitem[Brown \& Sandholm(2017{\natexlab{a}})Brown and
  Sandholm]{Brown17:Reduced}
Brown, N. and Sandholm, T.
\newblock Reduced space and faster convergence in imperfect-information games
  via pruning.
\newblock In \emph{International Conference on Machine Learning (ICML)},
  2017{\natexlab{a}}.

\bibitem[Brown \& Sandholm(2017{\natexlab{b}})Brown and
  Sandholm]{Brown17:Superhuman}
Brown, N. and Sandholm, T.
\newblock Superhuman {AI} for heads-up no-limit poker: {Libratus} beats top
  professionals.
\newblock \emph{Science}, pp.\  eaao1733, Dec. 2017{\natexlab{b}}.

\bibitem[Brown \& Sandholm(2019{\natexlab{a}})Brown and
  Sandholm]{Brown19:Solving}
Brown, N. and Sandholm, T.
\newblock Solving imperfect-information games via discounted regret
  minimization.
\newblock In \emph{Proceedings of the AAAI Conference on Artificial
  Intelligence}, volume~33, pp.\  1829--1836, 2019{\natexlab{a}}.

\bibitem[Brown \& Sandholm(2019{\natexlab{b}})Brown and
  Sandholm]{Brown19:Superhuman}
Brown, N. and Sandholm, T.
\newblock Superhuman {AI} for multiplayer poker.
\newblock \emph{Science}, 365\penalty0 (6456):\penalty0 885--890,
  2019{\natexlab{b}}.
\newblock ISSN 0036-8075.
\newblock \doi{10.1126/science.aay2400}.
\newblock URL \url{https://science.sciencemag.org/content/365/6456/885}.

\bibitem[Brown et~al.(2015)Brown, Ganzfried, and
  Sandholm]{Brown15:Hierarchical}
Brown, N., Ganzfried, S., and Sandholm, T.
\newblock Hierarchical abstraction, distributed equilibrium computation, and
  post-processing, with application to a champion no-limit {T}exas {H}old'em
  agent.
\newblock In \emph{International Conference on Autonomous Agents and
  Multi-Agent Systems (AAMAS)}, 2015.

\bibitem[Brown et~al.(2017)Brown, Kroer, and Sandholm]{Brown17:Dynamic}
Brown, N., Kroer, C., and Sandholm, T.
\newblock Dynamic thresholding and pruning for regret minimization.
\newblock In \emph{AAAI Conference on Artificial Intelligence (AAAI)}, 2017.

\bibitem[Brown et~al.(2018)Brown, Sandholm, and Amos]{Brown18:Depth}
Brown, N., Sandholm, T., and Amos, B.
\newblock Depth-limited solving for imperfect-information games.
\newblock \emph{arXiv preprint arXiv:1805.08195}, 2018.

\bibitem[Burch et~al.(2014)Burch, Johanson, and Bowling]{Burch14:Solving}
Burch, N., Johanson, M., and Bowling, M.
\newblock Solving imperfect information games using decomposition.
\newblock In \emph{AAAI Conference on Artificial Intelligence (AAAI)}, 2014.

\bibitem[Chambolle \& Pock(2011)Chambolle and Pock]{Chambolle11:First}
Chambolle, A. and Pock, T.
\newblock A first-order primal-dual algorithm for convex problems with
  applications to imaging.
\newblock \emph{Journal of Mathematical Imaging and Vision}, 2011.

\bibitem[Chambolle \& Pock(2016)Chambolle and Pock]{Chambolle16:Ergodic}
Chambolle, A. and Pock, T.
\newblock On the ergodic convergence rates of a first-order primal--dual
  algorithm.
\newblock \emph{Mathematical Programming}, 159\penalty0 (1-2):\penalty0
  253--287, 2016.

\bibitem[Chiang et~al.(2012)Chiang, Yang, Lee, Mahdavi, Lu, Jin, and
  Zhu]{Chiang12:Online}
Chiang, C.-K., Yang, T., Lee, C.-J., Mahdavi, M., Lu, C.-J., Jin, R., and Zhu,
  S.
\newblock Online optimization with gradual variations.
\newblock In \emph{Conference on Learning Theory}, pp.\  6--1, 2012.

\bibitem[Farina et~al.(2019)Farina, Kroer, Brown, and
  Sandholm]{Farina19:Stable}
Farina, G., Kroer, C., Brown, N., and Sandholm, T.
\newblock Stable-predictive optimistic counterfactual regret minimization.
\newblock In \emph{International Conference on Machine Learning (ICML)}, 2019.

\bibitem[Ganzfried \& Sandholm(2014)Ganzfried and
  Sandholm]{Ganzfried14:Potential}
Ganzfried, S. and Sandholm, T.
\newblock Potential-aware imperfect-recall abstraction with earth mover's
  distance in imperfect-information games.
\newblock In \emph{AAAI Conference on Artificial Intelligence (AAAI)}, 2014.

\bibitem[Ganzfried \& Sandholm(2015)Ganzfried and
  Sandholm]{Ganzfried15:Endgame}
Ganzfried, S. and Sandholm, T.
\newblock Endgame solving in large imperfect-information games.
\newblock In \emph{International Conference on Autonomous Agents and
  Multi-Agent Systems (AAMAS)}, 2015.

\bibitem[Gilpin \& Sandholm(2007)Gilpin and Sandholm]{Gilpin07:Lossless}
Gilpin, A. and Sandholm, T.
\newblock Lossless abstraction of imperfect information games.
\newblock \emph{Journal of the ACM}, 54\penalty0 (5), 2007.

\bibitem[Hazan(2016)]{Hazan16:Introduction}
Hazan, E.
\newblock Introduction to online convex optimization.
\newblock \emph{Foundations and Trends in Optimization}, 2\penalty0
  (3-4):\penalty0 157--325, 2016.

\bibitem[Hoda et~al.(2010)Hoda, Gilpin, Pe{\~n}a, and
  Sandholm]{Hoda10:Smoothing}
Hoda, S., Gilpin, A., Pe{\~n}a, J., and Sandholm, T.
\newblock Smoothing techniques for computing {N}ash equilibria of sequential
  games.
\newblock \emph{Mathematics of Operations Research}, 35\penalty0 (2), 2010.

\bibitem[Kroer(2019)]{Kroer19:First}
Kroer, C.
\newblock First-order methods with increasing iterate averaging for solving
  saddle-point problems.
\newblock \emph{arXiv preprint arXiv:1903.10646}, 2019.

\bibitem[Kroer \& Sandholm(2014)Kroer and Sandholm]{Kroer14:Extensive}
Kroer, C. and Sandholm, T.
\newblock Extensive-form game abstraction with bounds.
\newblock In \emph{Proceedings of the ACM Conference on Economics and
  Computation (EC)}, 2014.

\bibitem[Kroer \& Sandholm(2016)Kroer and Sandholm]{Kroer16:Imperfect}
Kroer, C. and Sandholm, T.
\newblock Imperfect-recall abstractions with bounds in games.
\newblock In \emph{Proceedings of the ACM Conference on Economics and
  Computation (EC)}, 2016.

\bibitem[Kroer et~al.(2015)Kroer, Waugh, K{\i}l{\i}n\c{c}-Karzan, and
  Sandholm]{Kroer15:Faster}
Kroer, C., Waugh, K., K{\i}l{\i}n\c{c}-Karzan, F., and Sandholm, T.
\newblock Faster first-order methods for extensive-form game solving.
\newblock In \emph{Proceedings of the ACM Conference on Economics and
  Computation (EC)}, 2015.

\bibitem[Kroer et~al.(2018{\natexlab{a}})Kroer, Farina, and
  Sandholm]{Kroer18:Solving}
Kroer, C., Farina, G., and Sandholm, T.
\newblock Solving large sequential games with the excessive gap technique.
\newblock In \emph{Proceedings of the Annual Conference on Neural Information
  Processing Systems (NIPS)}, 2018{\natexlab{a}}.

\bibitem[Kroer et~al.(2018{\natexlab{b}})Kroer, Waugh,
  K{\i}l{\i}n{\c{c}}-Karzan, and Sandholm]{Kroer18:Faster}
Kroer, C., Waugh, K., K{\i}l{\i}n{\c{c}}-Karzan, F., and Sandholm, T.
\newblock Faster algorithms for extensive-form game solving via improved
  smoothing functions.
\newblock \emph{Mathematical Programming}, pp.\  1--33, 2018{\natexlab{b}}.

\bibitem[Kuhn(1950)]{Kuhn50:Simplified}
Kuhn, H.~W.
\newblock A simplified two-person poker.
\newblock In Kuhn, H.~W. and Tucker, A.~W. (eds.), \emph{Contributions to the
  Theory of Games}, volume~1 of \emph{Annals of Mathematics Studies, 24}, pp.\
  97--103. Princeton University Press, Princeton, New Jersey, 1950.

\bibitem[Lanctot et~al.(2012)Lanctot, Gibson, Burch, Zinkevich, and
  Bowling]{Lanctot12:No}
Lanctot, M., Gibson, R., Burch, N., Zinkevich, M., and Bowling, M.
\newblock No-regret learning in extensive-form games with imperfect recall.
\newblock In \emph{International Conference on Machine Learning (ICML)}, 2012.

\bibitem[Moravcik et~al.(2016)Moravcik, Schmid, Ha, Hladik, and
  Gaukrodger]{Moravcik16:Refining}
Moravcik, M., Schmid, M., Ha, K., Hladik, M., and Gaukrodger, S.
\newblock Refining subgames in large imperfect information games.
\newblock In \emph{AAAI Conference on Artificial Intelligence (AAAI)}, 2016.

\bibitem[Morav{\v c}{\'\i}k et~al.(2017)Morav{\v c}{\'\i}k, Schmid, Burch,
  Lis{\'y}, Morrill, Bard, Davis, Waugh, Johanson, and
  Bowling]{Moravvcik17:DeepStack}
Morav{\v c}{\'\i}k, M., Schmid, M., Burch, N., Lis{\'y}, V., Morrill, D., Bard,
  N., Davis, T., Waugh, K., Johanson, M., and Bowling, M.
\newblock Deepstack: Expert-level artificial intelligence in heads-up no-limit
  poker.
\newblock \emph{Science}, 356\penalty0 (6337), May 2017.

\bibitem[Nemirovski(2004)]{Nemirovski04:Prox}
Nemirovski, A.
\newblock Prox-method with rate of convergence {O}(1/t) for variational
  inequalities with {Lipschitz} continuous monotone operators and smooth
  convex-concave saddle point problems.
\newblock \emph{SIAM Journal on Optimization}, 15\penalty0 (1), 2004.

\bibitem[Rakhlin \& Sridharan(2013{\natexlab{a}})Rakhlin and
  Sridharan]{Rakhlin13:Online}
Rakhlin, A. and Sridharan, K.
\newblock Online learning with predictable sequences.
\newblock In \emph{Conference on Learning Theory}, pp.\  993--1019,
  2013{\natexlab{a}}.

\bibitem[Rakhlin \& Sridharan(2013{\natexlab{b}})Rakhlin and
  Sridharan]{Rakhlin13:Optimization}
Rakhlin, S. and Sridharan, K.
\newblock Optimization, learning, and games with predictable sequences.
\newblock In \emph{Advances in Neural Information Processing Systems}, pp.\
  3066--3074, 2013{\natexlab{b}}.

\bibitem[Shalev-Shwartz \& Singer(2007)Shalev-Shwartz and
  Singer]{Schwartz07:Primal}
Shalev-Shwartz, S. and Singer, Y.
\newblock A primal-dual perspective of online learning algorithms.
\newblock \emph{Machine Learning}, 69\penalty0 (2-3):\penalty0 115--142, 2007.

\bibitem[Southey et~al.(2005)Southey, Bowling, Larson, Piccione, Burch,
  Billings, and Rayner]{Southey05:Bayes}
Southey, F., Bowling, M., Larson, B., Piccione, C., Burch, N., Billings, D.,
  and Rayner, C.
\newblock {Bayes}' bluff: Opponent modelling in poker.
\newblock In \emph{Proceedings of the 21st Annual Conference on Uncertainty in
  Artificial Intelligence (UAI)}, July 2005.

\bibitem[Syrgkanis et~al.(2015)Syrgkanis, Agarwal, Luo, and
  Schapire]{Syrgkanis15:Fast}
Syrgkanis, V., Agarwal, A., Luo, H., and Schapire, R.~E.
\newblock Fast convergence of regularized learning in games.
\newblock In \emph{Advances in Neural Information Processing Systems}, pp.\
  2989--2997, 2015.

\bibitem[Tammelin et~al.(2015)Tammelin, Burch, Johanson, and
  Bowling]{Tammelin15:Solving}
Tammelin, O., Burch, N., Johanson, M., and Bowling, M.
\newblock Solving heads-up limit {T}exas hold'em.
\newblock In \emph{Proceedings of the 24th International Joint Conference on
  Artificial Intelligence (IJCAI)}, 2015.

\bibitem[{von Stengel}(1996)]{Stengel96:Efficient}
{von Stengel}, B.
\newblock Efficient computation of behavior strategies.
\newblock \emph{Games and Economic Behavior}, 14\penalty0 (2):\penalty0
  220--246, 1996.

\bibitem[Zinkevich(2003)]{Zinkevich03:Online}
Zinkevich, M.
\newblock Online convex programming and generalized infinitesimal gradient
  ascent.
\newblock In \emph{International Conference on Machine Learning (ICML)}, pp.\
  928--936, Washington, DC, USA, 2003.

\bibitem[Zinkevich et~al.(2007)Zinkevich, Bowling, Johanson, and
  Piccione]{Zinkevich07:Regret}
Zinkevich, M., Bowling, M., Johanson, M., and Piccione, C.
\newblock Regret minimization in games with incomplete information.
\newblock In \emph{Proceedings of the Annual Conference on Neural Information
  Processing Systems (NIPS)}, 2007.

\end{thebibliography}

\clearpage
\appendix
\section{Proofs: Optimistic Follow-the-Regularized-Leader}

\subsection{Continuity of the Argmin-Function}
Intuitively, the role of the regularizer $d$ is to \emph{smooth out} the linear objective function $\langle \cdot, \vec{L}\rangle$. So, it seems only reasonable to expect that, the higher the constant that multiplies $d$, the less the argmin $\tilde{x}(\vec{L}) = \argmin_{\vec{x} \in \cX} \langle \vec{L}, \vec{x} \rangle + \frac{1}{\eta}d(\vec{x})$ is affected by small changes in $\vec{L}$. In fact, the following holds:

\begin{restatable}{lemma}{xtildelipschitz}\label{lem:x tilde lipschitz}
  Let $d$ be 1-strongly convex with respect to a norm $\|\cdot\|$. The argmin-function $\tilde{x}$ is $\eta$-Lipschitz continuous with respect to the dual norm $\|\cdot\|_\ast$, that is
  \[
    \|\tilde{x}(\vec{L}) - \tilde{x}(\vec{L}')\| \le \eta \|\vec{L} - \vec{L}'\|_*.
  \]
\end{restatable}
\begin{proof}
  The variational inequality for the optimality of $\tilde{x}(\vec{L})$ implies
  \begin{equation}\label{eq:vi1}
    \left\langle \vec{L} + \frac{1}{\eta}\nabla d(\tilde{x}(\vec{L})), \tilde{x}(\vec{L}') - \tilde{x}(\vec{L}) \right\rangle \ge 0.
  \end{equation}
  Symmetrically for $\tilde{x}(\vec{L}')$, we find that
  \begin{equation}\label{eq:vi2}
    \left\langle \vec{L}' + \frac{1}{\eta}d(\tilde{x}(\vec{L}')), \tilde{x}(\vec{L}) - \tilde{x}(\vec{L}') \right\rangle \ge 0.
  \end{equation}
  Summing inequalities~\eqref{eq:vi1} and~\eqref{eq:vi2}, we obtain
  \begin{align*}
    &\frac{1}{\eta} \left\langle \nabla d(\tilde{x}(\vec{L})) - \nabla d(\tilde{x}(\vec{L}')), \tilde{x}(\vec{L}) - \tilde{x}(\vec{L}') \right\rangle \le
    \left\langle \vec{L}' - \vec{L}, \tilde{x}(\vec{L}) - \tilde{x}(\vec{L}') \right\rangle.
  \end{align*}
  Using the 1-strong convexity of $d(\cdot)$ on the left-hand side and the generalized Cauchy-Schwarz inequality on the right-hand side, we obtain
  \begin{equation*}
    \frac{1}{\eta}\|\tilde{x}(\vec{L}) - \tilde{x}(\vec{L}')\|^2 \le \|\tilde{x}(\vec{L}) - \tilde{x}(\vec{L}')\|\,\|\vec{L} - \vec{L}'\|_\ast,
  \end{equation*}
  and dividing by $\|\tilde{x}(\vec{L}) - \tilde{x}(\vec{L}')\|$ we obtain the Lipschitz continuity of the argmin-function $\tilde{x}$.
\end{proof}

Another observation that will be crucial in the analysis is that the objective function $\langle \vec{L}, \vec{x}\rangle + (1/\eta) d(\vec{x})$ is $(1/\eta)$-strongly convex. Hence, for all $\vec{L}$ and $\hat{\vec{x}}\in \cX$,
  \begin{equation}\label{eq:strong convexity of obj}
    \left(\langle \vec{L}, \hat{\vec{x}} \rangle + \frac{1}{\eta} d(\hat{\vec{x}})\right) \ge \left(\langle\vec{L}, \tilde{x}(\vec{L})\rangle + \frac{1}{\eta} d(\tilde{x}(\vec{L}))\right) + \frac{1}{2\eta} \|\hat{\vec{x}} - \tilde{x}(\vec{L})\|^2.
  \end{equation}

\subsubsection{The Omniscient Case}
The RVU property (Definition~\ref{def:rvu}) implies that if $\vec{m}^t = \vec{\ell}^t$ for all $t=1,\dots, T$, then the regret cumulated is bounded above by a constant, independent on the time horizon $T$, and can only go down over time. We now show that this indeed holds for OFTRL. In what follows, we will use the notation $\vec{L}^t$ to denote $\vec{L}^t \defeq \sum_{\tau=1}^t \vec{\ell}^\tau$ for all $t \ge 1$, and $\vec{L}^0 \defeq \vec{0}$.
\begin{lemma}\label{lem:oftrl omniscient}
  Let $T > 0$, and assume that OFTRL is set up so that $\vec{\ell}^t = \vec{m}^t$ (i.e., the prediction is \emph{omniscient}) for all $t = 1,\dots, T$. Furthermore, denote $\vec{x}_o^t \defeq \tilde{x}(\vec{L}^t)$ the decisions produces by OFTRL at all times $t \ge 0$. Then, the regret against any strategy $\hat{\vec{x}}\in \cX$ is bounded above as
  \[
    R^T(\hat{\vec{x}}) \defeq \sum_{t=1}^T \langle \vec{\ell}^t, \vec{x}_o^t - \hat{\vec{x}} \rangle \le \frac{1}{\eta}\big(d(\hat{\vec{x}}) - d(\vec{x}_o^0)\big) - \frac{1}{2\eta} \sum_{t=1}^T \|\vec{x}_o^t - \vec{x}_o^{t-1}\|^2.
  \]
  As a direct consequence, if OFTRL is fed with exact predictions, the cumulated regret is bounded as
  \[
    R^T = \max_{\hat{\vec{x}} \in \cX} R^T(\hat{\vec{x}}) \le \frac{\Delta_d}{\eta} - \frac{1}{2\eta} \sum_{t=1}^T \|\vec{x}_o^t - \vec{x}_o^{t-1}\|^2,
  \]
  where $\Delta_d \defeq \max_{\vec{x},\vec{y}\in\cX} \{d(\vec{x}) - d(\vec{y})\}$.
\end{lemma}
\begin{proof}

  Equation~\eqref{eq:strong convexity of obj} immediately implies that for all $\hat{\vec{x}} \in \cX$ and time $t$,
  \begin{equation}\label{eq:var ineq omniscient 2}
    \langle \vec{L}^t, \hat{\vec{x}} - \vec{x}_o^t \rangle \ge -\frac{1}{\eta} \big(d(\hat{\vec{x}})-d(\vec{x}_o^t)\big) + \frac{1}{2\eta} \|\hat{\vec{x}} - \vec{x}_o^t\|^2.
  \end{equation}
  Consequently, for all $\hat{\vec{x}} \in \cX$ we have
  \begin{align}
    \langle \vec{L}^T, \vec{x}_o^T - \hat{\vec{x}} \rangle &= \langle \vec{\ell}^T, \vec{x}_o^T - \hat{\vec{x}} \rangle + \langle \vec{L}^{T-1},\vec{x}_o^{T-1} - \hat{\vec{x}}\rangle + \langle \vec{L}^{T-1}, \vec{x}_o^{T} - \vec{x}_o^{T-1} \rangle\nonumber\\
    &\ge \langle \vec{L}^{T-1}, \vec{x}_o^{T-1}-\hat{\vec{x}}\rangle + \langle \vec{\ell}^T, \vec{x}_o^T - \hat{\vec{x}}  \rangle-\frac{1}{\eta}(d(\vec{x}_o^{T}) - d(\vec{x}_o^{T-1}))+\frac{1}{2\eta}\| \vec{x}_o^T - \vec{x}_o^{T-1}\|^2,\label{eq:omniscient induction}
\end{align}
where the first inequality is by~\eqref{eq:var ineq omniscient 2} applied to $\hat{\vec{x}} = \vec{x}^{T}$ and $t={T-1}$.
Using~\eqref{eq:omniscient induction} inductively (recursively expanding terms of the form $\langle \vec{L}^t, \vec{x}_o^t - \hat{\vec{x}} \rangle$ for all $t = T-1, T-2, \dots, 1$), we obtain
\begin{align*}
    \langle \vec{L}^T, \vec{x}_o^T - \hat{\vec{x}} \rangle &\ge \left(\sum_{t=1}^T \langle \vec{\ell}^t, \vec{x}_o^t - \hat{\vec{x}}\rangle\right) - \frac{1}{\eta}\big(d(\vec{x}_o^T) - d(\vec{x}_o^0)\big) + \frac{1}{2\eta}\sum_{t=1}^T \| \vec{x}_o^{t} - \vec{x}_o^{t-1}\|^2 ,
  \end{align*}
  Finally, inverting the sides of the inequality, we find
\begin{align*}
    R^T(\hat{\vec{x}}) &= \sum_{t=1}^T \langle\vec{\ell}^t, \vec{x}_o^t -\hat{\vec{x}}\rangle\\
    &\le \langle \vec{L}^T, \vec{x}_o^T - \hat{\vec{x}} \rangle + \frac{1}{\eta}\big(d(\vec{x}_o^T) - d(\vec{x}_o^0)\big) - \frac{1}{2\eta}\sum_{t=1}^T \| \vec{x}_o^{t} - \vec{x}_o^{t-1}\|^2\\
    &\le \frac{1}{\eta}\big(d(\hat{\vec{x}}) - d(\vec{x}_o^T) \big) + \frac{1}{\eta}\big(d(\vec{x}_o^T) - d(\vec{x}_o^0)\big) - \frac{1}{2\eta}\sum_{t=1}^T \| \vec{x}_o^{t} - \vec{x}_o^{t-1}\|^2 \\
    &= \frac{1}{\eta}\big(d(\hat{\vec{x}}) - d(\vec{x}_o^0)\big) - \frac{1}{2\eta}\sum_{t=1}^T \| \vec{x}_o^{t} - \vec{x}_o^{t-1}\|^2
  \end{align*}
  where the second inequality is by~\eqref{eq:var ineq omniscient 2} applied to $t = T$.
\end{proof}
\subsection{Bridging the Gap: The Non-Omniscient Case}
We are now ready to complete the proof of Theorem~\ref{thm:oftrl}. The idea of the proof is to use the Lipschitz continuity of the argmin-function (Lemma~\ref{lem:x tilde lipschitz}) to show that the regret generated by the non-omniscience of the OFTRL (that is, the fact that $\vec{m}^t \neq \vec{\ell}^t$) is proportional to the distance $\sum_t \|\vec{m}^t - \vec{\ell}^t\|_\ast^2$. A naive attempt would could be as follows:
\begin{align*}
  \langle \vec{\ell}^t, \vec{x}^t \rangle &= \langle \vec{\ell}^t, \tilde{x}(\vec{L}^t) \rangle + \langle \vec{\ell}^t, \tilde{x}(\vec{L}^{t-1} + \vec{m}^t) - \tilde{x}(\vec{L}^t)\rangle\\
    &\le \langle \vec{\ell}^t, \tilde{x}(\vec{L}^t)\rangle + {\eta \|\vec{\ell}^t\|_\ast} \|\vec{\ell}^t - \vec{m}^t\|_\ast,
\end{align*}
where the inequality follows from from the generalized Cauchy-Schwarz inequality and from Lemma~\ref{lem:x tilde lipschitz}. Unfortunately, the approach above is not powerful enough, as it would imply that the increase in regret is proportional to $\sum_t \|\vec{m}^t-\vec{\ell}^t\|_\ast$ (notice the different exponent). In order to obtain the better bound, we write instead
\begin{align}
    \langle \vec{\ell}^t, \vec{x}^t \rangle &= \langle \vec{\ell}^t, \tilde{x}(\vec{L}^t) \rangle + \langle \vec{\ell}^t - \vec{m}^t, \tilde{x}(\vec{L}^{t-1} + \vec{m}^t) - \tilde{x}(\vec{L}^t)\rangle + \langle \vec{m}^t, \tilde{x}(\vec{L}^{t-1} + \vec{m}^t) - \tilde{x}(\vec{L}^t)\rangle\nonumber\\
    &\le \langle \vec{\ell}^t, \tilde{x}(\vec{L}^t)\rangle + \eta\|\vec{\ell}^t - \vec{m}^t\|^2_\ast + \langle \vec{m}^t, \tilde{x}(\vec{L}^{t-1} + \vec{m}^t) - \tilde{x}(\vec{L}^t)\rangle\label{eq:omniscient completion}.
\end{align}
The first term corresponds to the case of omniscient predictions, which was already analyzed in the previous section. Hence, we are left with the task of bounding the loss from the last term. To this end, we will use the following lemma:
\begin{lemma}\label{lem:technical}
  Let $\vec{x}^t \defeq \tilde{x}(\vec{L}^{t-1} + \vec{m}^t)$ and $\vec{x}_o^t \defeq \tilde{x}(\vec{L}^t)$.
  For all $T$,
  \[
    \sum_{t=1}^T \langle \vec{m}^t, \vec{x}^t - \vec{x}_o^t\rangle \le  -R_o^T + \frac{\Delta_d}{\eta}- \frac{1}{4\eta} \sum_{t=1}^{T-1} \|\vec{x}^{t+1} - \vec{x}^t\|^2
  \]
  where
  \[
    R^T_o \defeq \max_{\hat{\vec{x}}\in\cX}\sum_{t=1}^T \langle \vec{\ell}^t, \vec{x}_o^t - \hat{\vec{x}} \rangle
  \]
  is the regret of the omniscient minimizer.
\end{lemma}
\begin{proof}
  We will make use of the following two inequalities, which are direct applications of~\eqref{eq:strong convexity of obj} where $L$ is set to $\vec{L}^{t-1} + \vec{m}^t$ and $\vec{L}^{t-1}$, respectively, and $\hat{\vec{x}}$ is set to $\vec{x}^t_o$ and $\vec{x}^t$, respectively:
\begin{align}
   \langle \vec{L}^{t-1} + \vec{m}^t, \vec{x}^t\rangle &\le \langle \vec{L}^{t-1} + \vec{m}^t, \vec{x}^t_o\rangle - \frac{1}{\eta}\big( d(\vec{x}^t) - d(\vec{x}^t_o) \big) - \frac{1}{2\eta} \|\vec{x}^t - \vec{x}^t_o\|^2, \label{eq:sub1}\\
   \langle \vec{L}^{t-1}, \vec{x}^{t-1}_o \rangle &\le \langle \vec{L}^{t-1}, \vec{x}^{t} \rangle - \frac{1}{\eta}\big( d(\vec{x}^{t-1}_o) - d(\vec{x}^t) \big) - \frac{1}{2\eta} \|\vec{x}^t - \vec{x}^{t-1}_o\|^2 \label{eq:sub2}.
\end{align}
With the above, we can write, for all timesteps $t$,
\begin{align*}
  \langle \vec{m}^t, \vec{x}^t - \vec{x}^t_o \rangle &= -\langle \vec{m}^t, \vec{x}^t_o \rangle  + \langle \vec{L}^{t-1} + \vec{m}^t, \vec{x}^t\rangle - \langle \vec{L}^{t-1}, \vec{x}^t\rangle \\
      &\le -\langle \vec{m}^t, \vec{x}_o^t\rangle + \langle \vec{L}^{t-1} + \vec{m}^t, \vec{x}^t\rangle - \langle \vec{L}^{t-1}, \vec{x}^{t-1}_o \rangle - \frac{1}{\eta}\big( d(\vec{x}^{t-1}_o) - d(\vec{x}^t) \big) - \frac{1}{2\eta} \|\vec{x}^t - \vec{x}^{t-1}_o\|^2\\
      &\le -\langle \vec{m}^t, \vec{x}_o^t\rangle + \langle \vec{L}^{t-1} + \vec{m}^t, \vec{x}^t_o\rangle- \langle \vec{L}^{t-1}, \vec{x}^{t-1}_o \rangle - \frac{1}{\eta}\big( d(\vec{x}^t) - d(\vec{x}^t_o) \big) - \frac{1}{\eta}\big( d(\vec{x}^{t-1}_o) - d(\vec{x}^t) \big) \\
               &\hspace{8.5cm}- \frac{1}{2\eta} \|\vec{x}^t - \vec{x}^t_o\|^2   - \frac{1}{2\eta} \|\vec{x}^t - \vec{x}^{t-1}_o\|^2\\
      &= \langle \vec{L}^{t-1}, \vec{x}_o^t - \vec{x}_o^{t-1}\rangle + \frac{1}{\eta}\big( d(\vec{x}_o^t) - d(\vec{x}_o^{t-1})\big) - \frac{1}{2\eta} \|\vec{x}^t - \vec{x}^t_o\|^2   - \frac{1}{2\eta} \|\vec{x}^t - \vec{x}^{t-1}_o\|^2,
\end{align*}
where the first inequality follows from~\eqref{eq:sub2} and the second inequality follows from~\eqref{eq:sub1}. Note that
\begin{align*}
      \langle \vec{L}^{t-1}, \vec{x}_o^t - \vec{x}_o^{t-1}\rangle &= \langle \vec{L}^{t}, \vec{x}_o^t \rangle - \langle \vec{L}^{t-1}, \vec{x}_o^{t-1} \rangle - \langle \vec{\ell}^t, \vec{x}_o^t \rangle.
\end{align*}
Hence,
\[
\langle \vec{m}^t, \vec{x}^t - \vec{x}^t_o \rangle \le \langle \vec{L}^{t}, \vec{x}_o^t \rangle - \langle \vec{L}^{t-1}, \vec{x}_o^{t-1} \rangle - \langle \vec{\ell}^t, \vec{x}^t \rangle +\frac{1}{\eta}\big( d(\vec{x}_o^t) - d(\vec{x}_o^{t-1})\big) - \frac{1}{2\eta} \|\vec{x}^t - \vec{x}^t_o\|^2   - \frac{1}{2\eta} \|\vec{x}^t - \vec{x}^{t-1}_o\|^2.
\]
Finally, summing over all $t$, we obtain
\begin{align}
    \sum_{t=1}^T \langle \vec{m}^t, \vec{x}^t - \vec{x}^t_o\rangle &\le \langle \vec{L}^T, \vec{x}_o^T \rangle - \sum_{t=1}^T \langle \vec{\ell}^t, \vec{x}_o^t\rangle + \frac{1}{\eta}\big( d(\vec{x}_o^T) - d(\vec{x}_o^{0})\big) - \frac{1}{2\eta} \bigg(\sum_{t=1}^T \|\vec{x}^t - \vec{x}^{t}_o\|^2 + \sum_{t=1}^T \|\vec{x}_o^t - \vec{x}^{t-1}_o\|^2\bigg)\nonumber\\
    &\le  \langle \vec{L}^T, \vec{x}_o^T \rangle - \sum_{t=1}^T \langle \vec{\ell}^t, \vec{x}_o^t\rangle +\frac{1}{\eta} \big( d(\vec{x}^T_o) - d(\vec{x}^0_o) \big) - \frac{1}{2\eta} \bigg(\sum_{t=1}^{T-1} \|\vec{x}^t - \vec{x}^{t}_o\|^2 + \sum_{t=1}^{T-1} \|\vec{x}^{t+1} - \vec{x}^{t}_o\|^2\bigg)\nonumber\\
    &\le \langle \vec{L}^T, \vec{x}_o^T \rangle - \sum_{t=1}^T \langle \vec{\ell}^t, \vec{x}_o^t\rangle +\frac{1}{\eta} \big( d(\vec{x}^T_o) - d(\vec{x}^0_o) \big) - \frac{1}{4\eta} \sum_{t=1}^{T-1} \|\vec{x}^{t+1} - \vec{x}^t\|^2,\label{eq:almost done}
\end{align}
where the first inequality comes from simplifying telescopic sums, the second by removing some term from the last parenthesis, and the third from the parallelogram inequality
\[
  \|\vec{a}\|^2 + \|\vec{b}\|^2 \ge \frac{1}{2} \|\vec{a} - \vec{b}\|^2,
\]
which holds for all choice of vectors $\vec{a},\vec{b},\vec{c}$ and norm $\|\cdot\|$. The last step of the proof is to notice that by~\eqref{eq:var ineq omniscient 2}
\begin{align*}
  \langle \vec{L}^T, \vec{x}_o^T \rangle - \sum_{t=1}^T \langle \vec{\ell}^t, \vec{x}_o^t\rangle &\le \langle \vec{L}^T, \hat{\vec{x}} \rangle- \sum_{t=1}^T \langle \vec{\ell}^t, \vec{x}_o^t\rangle + \frac{1}{\eta}\big(d(\hat{\vec{x}}) - d(\vec{x}_o^T)\big)\\
   &= -R_o^T(\hat{\vec{x}}) + \frac{1}{\eta}\big(d(\hat{\vec{x}}) - d(\vec{x}_o^T)\big)
\end{align*}
for all $\hat{\vec{x}}\in\cX$. Substituting into~\eqref{eq:almost done} yields
\begin{align*}
  \sum_{t=1}^T \langle \vec{m}^t, \vec{x}^t - \vec{x}^t_o\rangle &\le -R_o^T(\hat{\vec{x}})+ \frac{1}{\eta}\big(d(\hat{\vec{x}}) - d(\vec{x}_o^0) \big) - \frac{1}{4\eta} \sum_{t=1}^{T-1} \|\vec{x}^{t+1} - \vec{x}^t\|^2\\
   &= -R_o^T + \frac{\Delta_d}{\eta}- \frac{1}{4\eta} \sum_{t=1}^{T-1} \|\vec{x}^{t+1} - \vec{x}^t\|^2,
\end{align*}
as we wanted to show.
\end{proof}
Finally, using Lemma~\ref{lem:technical} together with~\eqref{eq:omniscient completion} and Lemma~\ref{lem:oftrl omniscient}, we obtain the proof of Theorem~\ref{thm:oftrl}.

\thmoftrl*
\begin{proof}[Proof of Theorem~\ref{thm:oftrl}]
As we already mentioned, the fundamental idea of the proof is to bound how much regret is generated from by the fact that the predictions $\vec{m}^t$ are \emph{not} perfect. Hence, the first step is to relate the notion of regret to the fact that the decisions produced are $\vec{x}^t$ instead of $\vec{x}_o^t$:
\begin{align*}
  R^T &= \max_{\hat{\vec{x}} \in \cX} \left\{\sum_{t=1}^T \langle \vec{\ell}^t, \vec{x}^t - \hat{\vec{x}} \rangle \right\} \\
      &= \max_{\hat{\vec{x}} \in \cX} \left\{\sum_{t=1}^T \langle \vec{\ell}^t, \vec{x}_o^t - \hat{\vec{x}} \rangle \right\} + \sum_{t=1}^T \langle \vec{\ell}^t - \vec{m}^t, \vec{x}^t - \vec{x}_o^t \rangle + \sum_{t=1}^T \langle \vec{m}^t, \vec{x}^t - \vec{x}_o^t \rangle.
\end{align*}
The sum inside of the max is bounded according to the analysis of the omniscient case (Lemma~\ref{lem:oftrl omniscient}). Using the (generalized) Cauchy-Schwarz inequality and the Lipschitz continuity of $\tilde{x}$, the sum in the middle can be easily bounded as
\[
  \sum_{t=1}^T \langle \vec{\ell}^t - \vec{m}^t, \vec{x}^t - \vec{x}_o^t \rangle \le \eta \sum_{t=1}^T \|\vec{\ell}^t - \vec{m}^t\|_\ast^2.
\]
Finally, the third summation was studied in Lemma~\ref{lem:technical}. Putting everything together, we conclude that
\begin{align*}
  R^T &\le R_o^T + \eta\sum_{t=1}^T \|\vec{\ell}^t - \vec{m}^t\|_\ast^2 - R_o^T + \frac{\Delta_d}{\eta} - \frac{1}{4\eta}\sum_{t=1}^T \|\vec{x}^t - \vec{x}^{t-1}\|^2 \\
      &\le \frac{\Delta_d}{\eta} + \eta \sum_{t=1}^T \|\vec{\ell}^t - \vec{m}^t\|_\ast^2 - \frac{1}{4\eta}\sum_{t=1}^T \|\vec{x}^t - \vec{x}^{t-1}\|^2.
\end{align*}
\end{proof}

\section{Proofs: Dilated Entropy Distance Generating Function}
In the proof of Theorem~\ref{thm:strongly convex dgf}, we will make use of the following, useful fact.

\begin{restatable}{lemma}{lemdgfgradient}\label{lem:dgf gradient}
    Let $d(\vec{x}) \!=\! \sum_{j\in \cJ}\! x_{p_j} d_j(\vec{x}_j/x_{p_j})$ be a dilated DGF. Then, for all $\vec{x}$ we have
$
    \langle \nabla d(\vec{x}) , \vec{x} \rangle \!=\! d(\vec{x}).
$
\end{restatable}
\begin{proof}
    With some simple algebra, it is easy to show that the partial derivatives with respect to $x_{ja}$ is
    \[
      \frac{\partial d}{\partial x_{ja}}(\vec{x}) =
                           \frac{\partial  r_j}{\partial x_{ja}}\!\left(\frac{\vec{x}_j}{x_{p_j}}\right)
                         + \sum_{j' \in C_{ja}} d_{j'}\!\left(\frac{\vec{x}_{j'}}{x_{ja}}\right)
                         - \sum_{j' \in C_{ja}} \left\langle \nabla d_{j'}\!\left(\frac{\vec{x}_{j'}}{x_{ja}}\right)\!,\, \frac{\vec{x}_{j'}}{x_{ja}} \right\rangle,
    \]
    and that the partial derivative of $d$ with respect to $\vec{x}_\phi$ is
    \[
      \frac{\partial d}{\partial x_\phi}(\vec{x}) = d_{r}\!\left(\frac{\vec{x}_{r}}{x_\phi}\right) - \left\langle \nabla d_{r}\!\left(\frac{\vec{x}_{r}}{x_\phi}\right)\!,\, \frac{\vec{x}_{r}}{x_\phi} \right\rangle.
    \]
    As a consequence, for all $j\in\cJ$,
    \begin{align*}
      \left\langle\nabla_{\vec{x}_j} d(\vec{x}),\, \vec{x}_j\right\rangle &=
                           \left\langle \nabla d_j\!\left(\frac{\vec{x}_j}{x_{p_j}}\right), \vec{x}_j \right\rangle
                         - \sum_{a \in A_j}\sum_{j' \in C_{ja}} \left\langle \nabla d_{j'}\!\left(\frac{\vec{x}_{j'}}{x_{ja}}\right)\!,\, \vec{x}_{j'} \right\rangle
                         + \sum_{a \in A_j}\sum_{j' \in C_{ja}} x_{ja}\, d_{j'}\!\left(\frac{\vec{x}_{j'}}{x_{ja}}\right)
    \end{align*}
    Hence,
    \begin{align*}
      \langle \nabla d(\vec{x}), \vec{x} \rangle &= \frac{\partial d}{\partial x_{\phi}}(\vec{x}) + \sum_j \langle \nabla_{\vec{x}_j} d(\vec{x}),\, \vec{x}_j\rangle \\ &= d_{r}\!\left(\frac{\vec{x}_r}{x_\phi}\right) + \sum_{j\in\cJ} \sum_{a \in A_j}\sum_{j' \in C_{ja}} x_{ja}\, d_{j'}\!\left(\frac{\vec{x}_{j'}}{x_{ja}}\right) \\
&= \sum_{j\in\cJ} x_{p_j} d_{j}\!\left(\frac{\vec{x}_j}{x_{p_j}}\right) = d(\vec{x}),
    \end{align*}
    where the second inequality comes from the observation that the inner products with the gradients cancel out.
\end{proof}

\thmstronglconvexdgf*
\begin{proof}
Let $x,y$ be arbitrary vectors in $\in\cX$. Using the strong convexity of $d_j$ and the fact that $\vec{x} \in \cX$, we obtain for all $j\in\cJ$
\begin{align*}
 \sum_{a \in A_j}\sum_{j'\in C_{ja}} y_{ja}\, d_{j'}\!\left(\frac{\vec{y}_{j'}}{y_{ja}}\right)
                    &\ge \sum_{a \in A_j}\sum_{j'\in C_{ja}}
                                y_{ja}\cdot\left(   d_{j'}\!\left(\frac{\vec{x}_{j'}}{x_{ja}}\right)
                                        +\left\langle \nabla d_{j'}\!\left(\frac{\vec{x}_{j'}}{x_{ja}}\right),\frac{\vec{y}_{j'}}{y_{ja}}-\frac{\vec{x}_{j'}}{x_{ja}}\right\rangle
                                        +\frac{\mu_{j'}}{2}\left\|\frac{\vec{y}_{j'}}{y_{ja}}-\frac{\vec{x}_{j'}}{x_{ja}}\right\|_2^2
                                \right)\\
                    &=
                      \left(
                         \sum_{a \in A_j}\sum_{j'\in C_{ja}} y_{ja}\, d_{j'}\!\left(\frac{\vec{x}_{j'}}{x_{ja}}\right)
                         - \sum_{a \in A_j}\sum_{j' \in C_{ja}}y_{ja} \left\langle \nabla d_{j'}\!\left(\frac{\vec{x}_{j'}}{x_{ja}}\right)\!,\, \frac{\vec{x}_{j'}}{x_{ja}}\right\rangle
                      \right)\\[2mm]
                   &\hspace{1cm}
                                 - \left\langle \nabla d_j\!\left(\frac{\vec{x}_j}{x_{p_j}}\right), \vec{y}_j \right\rangle
                                 + \sum_{a \in A_j}\sum_{j' \in C_{ja}} \left\langle \nabla d_{j'}\!\left(\frac{\vec{x}_{j'}}{x_{ja}}\right)\!,\, \vec{y}_{j'}\right\rangle\\
                   &\hspace{1cm}
                                 + \sum_{a \in A_j}\sum_{j'\in C_{ja}}\frac{\mu_{j'}}{2} y_{ja}\left\|\frac{\vec{y}_{j'}}{y_{ja}}-\frac{\vec{x}_{j'}}{x_{ja}}\right\|_2^2\\
                   &= \left\langle\nabla_{\vec{x}_j} d(\vec{x}),\, \vec{y}_j\right\rangle - \left\langle \nabla d_j\!\left(\frac{x_j}{x_{p_j}}\right), \vec{y}_j \right\rangle
                      + \sum_{a \in A_j}\sum_{j' \in C_{ja}} \left\langle \nabla d_{j'}\!\left(\frac{\vec{x}_{j'}}{x_{ja}}\right)\!,\, \vec{y}_{j'}\right\rangle\\
                   &\hspace{1cm}
                   + \sum_{a \in A_j}\sum_{j'\in C_{ja}}\frac{\mu_{j'}}{2} y_{ja}\left\|\frac{\vec{y}_{j'}}{y_{ja}}-\frac{\vec{x}_{j'}}{x_{ja}}\right\|_2^2.
\end{align*}
Summing over all $j$, we obtain
\begin{align}
  d(\vec{y}) &\ge \langle \nabla d(\vec{x}), \vec{y}\rangle + \sum_j \frac{\mu_{j}}{2} y_{p_j}\left\|\frac{\vec{y}_j}{y_{p_j}}-\frac{\vec{x}_j}{x_{p_j}}\right\|_2^2 \nonumber\\
        &= d(\vec{x}) + \langle \nabla d(\vec{x}), \vec{y} - \vec{x}\rangle + \sum_j \frac{\mu_{j}}{2} y_{p_j}\left\|\frac{\vec{y}_j}{y_{p_j}}-\frac{\vec{x}_j}{x_{p_j}}\right\|_2^2 \nonumber\\
        &\ge d(\vec{x}) + \langle \nabla d(\vec{x}), \vec{y} - \vec{x}\rangle + \sum_j \frac{\mu_{j}}{2} y^2_{p_j}\left\|\frac{\vec{y}_j}{y_{p_j}}-\frac{\vec{x}_j}{x_{p_j}}\right\|_2^2 , \label{eq:plugintome}
\end{align}
where the equality comes from Lemma~\ref{lem:dgf gradient} and the second   inequality from the fact that $y_{p_j} \in [0, 1]$. Note that for all $j$ and for all $\vec{x}_j, \vec{y}_j$,
\begin{align}
    \frac{1}{2}\left\|\vec{y}_j - \vec{x}_j\right\|_2^2
      &= \frac{1}{2}\left\|y_{p_j} \left(\frac{\vec{y}_j}{y_{p_j}} - \frac{\vec{x}_j}{x_{p_j}}\right) + (y_{p_j} - x_{p_j})\frac{\vec{x}_j}{x_{p_j}}\right\|_2^2  \nonumber\\
      &\le y_{p_j}^2 \left\|\frac{\vec{y}_j}{y_{p_j}} - \frac{\vec{x}_j}{x_{p_j}}\right\|_2^2 + (y_{p_j} - x_{p_j})^2\left\|\frac{\vec{x}_j}{x_{p_j}}\right\|_2^2 \nonumber\\
      &\le y_{p_j}^2 \left\|\frac{\vec{y}_j}{y_{p_j}} - \frac{\vec{x}_j}{x_{p_j}}\right\|_2 + (y_{p_j} - x_{p_j})^2, \label{eq:invertme}
\end{align}
where the first inequality is by parallelogram inequality, while the second inequality follows from the fact that $\vec{x}_j/x_{p_j}\in \Delta^{n_j}$ and therefore it has norm upper-bounded by 1. By inverting~\eqref{eq:invertme}, we obtain
\begin{align*}
    \sum_{j} \mu_j y_{p_j}^2 \left\|\frac{\vec{y}_j}{y_{p_j}} - \frac{\vec{x}_j}{x_{p_j}}\right\|_2^2
      &\ge \frac{1}{2}\sum_{j} \mu_j\|\vec{y}_j - \vec{x}_j\|^2 - \sum_{j} \mu_j (y_{p_j} - x_{p_j})^2\\
      &= \sum_{j} \sum_{a\in A_j} \left(\frac{\mu_j}{2} - \sum_{j' \in C_{ja}} \mu_{j'}\right) (y_{ja} - x_{ja})^2.
\end{align*}
Finally, plugging into~\eqref{eq:plugintome}
\[
\sum_{j} \mu_j y_{p_j} \left\|\frac{\vec{y}_j}{y_{p_j}} - \frac{\vec{x}_j}{x_{p_j}}\right\|_2^2
\ge
  \bar\sigma \sum_{j}\sum_{a\in A_j}(x_{ja} - y_{ja})^2 = \bar \sigma \|\vec{y} - \vec{x}\|_2^2,
\]
and we conclude that $d$ is $\bar\sigma$-strongly convex with respect to the Euclidean norm over $\cX$, like we wanted to show.
\end{proof}

\section{Proofs: Local Regret Minimization}
\propdecompositionlocalprox*
\begin{proof}
  We will show the decomposition assuming that $j$ is the root decision
point, the general case follows by induction. The prox mapping
\eqref{eq:sequence form prox} can be written as
  \begin{align*}
    \prox(\vec{g},\hat{\vec{x}}) &= \argmin_{\vec{x} \in \cX} \big\{ \langle \vec{g}, \vec{x} \rangle +
D(\vec{x} \dmid \hat{\vec{x}})\big\} \\
    &= \argmin_{\vec{x} \in \cX} \big\{ \langle \vec{g}, \vec{x} \rangle + d(\vec{x}) -
d(\hat{\vec{x}}) - \langle \nabla d(\hat{\vec{x}}), \vec{x} - \hat{\vec{x}} \rangle \big\}\\
    &= \argmin_{\vec{x} \in \cX} \big\{ \langle \vec{g} - \nabla d(\hat{\vec{x}}), \vec{x} \rangle +
d(\vec{x}) \big\}.
  \end{align*}
  Now we can use the fact that we are using a dilated DGF and the sequential
structure of SDM problems to write the problem of finding just the minimizer
for $j$ as follows
  \begin{align*}
    &\argmin_{\vec{b}_j \in \Delta^{n_j}} \bigg\{ \langle \vec{g}_j - \nabla_j d(\hat{\vec{x}}), \vec{b}_j \rangle
+ d_j(\vec{b}_j)  + \sum_{a \in A_j} \sum_{j' \in \childinfosets{j,a}}
b_{j,a}\bigg[\min_{\vec{x} \in \cX_{\subt{j'}}} \langle \vec{g}_{\subt{j'}} - \nabla
d_{\subt{j'}}(\hat{\vec{x}}_{\subt{j'}}), \vec{x} \rangle +
d_{\subt{j'}}(\vec{x}) \bigg] \bigg\}\\
    &=
    \argmin_{\vec{b}_j \in \Delta^{n_j}} \bigg\{\langle \vec{g}_j - \nabla_j d(\hat{\vec{x}}), \vec{b}_j \rangle
+ d_j(\vec{b}_j)  - \sum_{a \in A_j} \sum_{j' \in \childinfosets{j,a}}
b_{j,a} d^*_{\subt{j'}}(-\vec{g}_{\subt{j'}} + \nabla
d_{\subt{j'}}(\hat{\vec{x}}_{\subt{j'}}))\bigg\}.
  \end{align*}
  Now we note that the index of $\nabla_j d(\hat{\vec{x}})$ corresponding to
each $a \in A_j$ can be expanded as
  \[
    \nabla_{j,a} d(\hat{\vec{x}}) = \nabla_{a} d_j\left(\frac{\hat{\vec{x}}_j}{\xhat_{p_j}}\right) + \sum_{j' \in
\childinfosets{j,a}} \left( d_{j'}\left(\frac{\hat{\vec{x}}_{j'}}{\xhat_{p_{j'}}}\right) - \left\langle  \nabla
d_{j'}\left(\frac{\hat{\vec{x}}_{j'}}{\xhat_{p_{j'}}}\right), \frac{\hat{\vec{x}}_{j'}}{\xhat_{p_{j'}}} \right\rangle \right).
  \]
  Plugging this into our expression for the minimizer for $j$, and noting that we can add and remove $d_j\left(\frac{\hat{\vec{x}}_j}{\xhat_{p_j}}\right)$ without changing the $\argmin$, gives
  \begin{align*}
    & \argmin_{\vec{b}_j \in \Delta^{n_j}}\ \Bigg\{\langle \vec{g}_j , \vec{b}_j \rangle + D_j\bigg(\vec{b}_j \ \bigg\|\ \frac{\hat{\vec{x}}_j}{\xhat_{p_j}}\bigg)
- \sum_{a \in A_j} \sum_{j' \in \childinfosets{j,a}} b_{j,a} \Bigg[
d_{\subt{j'}}^*\big(-\vec{g}_{\subt{j'}} + \nabla d_{\subt{j'}}(\hat{\vec{x}}_{\subt{j'}})\big) -
d_{j'}\bigg(\frac{\hat{\vec{x}}_j}{\xhat_{p_j}}\bigg)\\
  & \hspace{9.5cm}+ \left\langle  \nabla
d_{j'}\left(\frac{\hat{\vec{x}}_{j'}}{\xhat_{p_{j'}}}\right), \frac{\hat{\vec{x}}_{j'}}{\xhat_{p_{j'}}} \right\rangle\Bigg]\Bigg\} \\
    &= \argmin_{\vec{b}_j \in \Delta^{n_j}}\ \langle \hat{\vec{g}}_j , \vec{b}_j \rangle +
      D_j\bigg(\vec{b}_j \ \bigg\|\ \frac{\hat{\vec{x}}_j}{\xhat_{p_j}}\bigg)
  \end{align*}
  as we wanted to show.
\end{proof}

\thmomdseparable*
\begin{proof}
  Note that the OMD update in \eqref{eq:omd} is a prox mapping as given in Proposition~\ref{prop:decomposition local prox}. The update is:
\begin{align}
  \vec{x}^{t+1} = \argmin_{\vec{x} \in \cX} \bigg\{ \langle \vec{\ell}^{t}, \vec{x} \rangle + \frac{1}{\eta} D(\vec{x} \dmid \vec{x}^{t})  \bigg\}
  \Rightarrow
  \vec{x}^{t+1}_j = \vec{x}^{t+1}_{p_j} \ \argmin_{\vec{b}_j \in \symp{n_j}} \bigg\{ \langle \hat{\vec{\ell}}^{t}, \vec{x} \rangle + \frac{1}{\eta} D_j\bigg(\vec{b}_j\ \bigg\|\ \frac{\vec{x}^{t}}{x^t_{p_j}}\bigg)  \bigg\}
  .
\end{align}
  But this is exactly the same as the OMD update resulting from running OMD on simplex $j$ with the modified loss, using DGF $d_j$ and its associated Bregman divergence. The same logic shows that Optimistic OMD is a local variant of itself, since it is a repeated  sequence of two prox mappings.
\end{proof} 
\end{document}